\documentclass[letterpaper,11pt]{article}
\usepackage{amsmath,amssymb,graphicx,latexsym,amsthm,xspace}

\newtheorem{theorem}{Theorem}[section]
\newtheorem{lemma}[theorem]{Lemma}

\newtheorem{fact}[theorem]{Fact}

\newtheorem{definition}[theorem]{Definition}
\newtheorem{observation}[theorem]{Observation}

\newcommand{\poly}{\mathrm{poly}}
\newcommand{\dtv}{d_{TV}}

\newcommand{\cald}{\mathcal{D}}
\newcommand{\calpsat}{\mathcal{P}_{\mathrm{sat}}}
\newcommand{\calpfar}{\mathcal{P}_{\mathrm{far}}}
\newcommand{\caldsat}{\mathcal{D}_{\mathrm{sat}}}
\newcommand{\caldfar}{\mathcal{D}_{\mathrm{far}}}
\newcommand{\calksat}{\mathcal{K}_{\mathrm{sat}}}
\newcommand{\calkfar}{\mathcal{K}_{\mathrm{far}}}
\newcommand{\caltsat}{\mathcal{T}_{\mathrm{sat}}}
\newcommand{\caltfar}{\mathcal{T}_{\mathrm{far}}}
\newcommand{\cy}{\mathrm{cy}}

\newcommand{\calo}{\mathcal{O}}
\newcommand{\cala}{\mathcal{A}}
\newcommand{\calg}{\mathcal{G}}

\newcommand{\calu}{\mathcal{U}}

\newcommand{\bfX}{\mathbf{X}}

\newcommand{\bfx}{\mathbf{x}}
\newcommand{\bfy}{\mathbf{y}}
\newcommand{\bfb}{\mathbf{b}}

\newcommand{\ci}{\perp\!\!\!\perp}
\newcommand{\bit}{\{0,1\}}
\newcommand{\supp}{\mathrm{Supp}}

\newcommand{\mislong}{\textsf{Maximum Independent Set}\xspace}
\newcommand{\mis}{\textsf{MIS}\xspace}

\newcommand{\kcsp}{\textsf{$k$-CSP}\xspace}
\newcommand{\csp}[1]{\textsf{CSP}(#1)\xspace}

\newcommand{\xor}{\textsf{XOR}\xspace}
\newcommand{\nae}{\textsf{NAE}\xspace}
\newcommand{\equ}{\textsf{EQU}\xspace}
\newcommand{\txor}{\textsf{$2$-XOR}\xspace}
\newcommand{\kxor}{$k$-\textsf{XOR}\xspace}
\newcommand{\kequ}{$k$-\textsf{EQU}\xspace}
\newcommand{\knae}{$k$-\textsf{NAE}\xspace}

\newcommand{\maxcsp}[1]{\textsf{Max CSP}(#1)\xspace}

\newcommand{\maxcut}{\textsf{Max Cut}\xspace}

\newcommand{\maxkcsp}{\textsf{Max $k$-CSP}\xspace}

\newcommand{\maxkxor}{\textsf{Max $k$-XOR}\xspace}

\title{Lower Bounds on Query Complexity for Testing Bounded-Degree CSPs}
\author{Yuichi Yoshida\thanks{This work was conducted while the author was visiting Rutgers University.}\\\\
  School of Informatics, Kyoto University, and\\ Preferred Infrastructure, Inc.\\yyoshida@lab2.kuis.kyoto-u.ac.jp}

\date{}

\begin{document}
\setcounter{page}{0}
\maketitle
\begin{abstract}
  In this paper, we consider lower bounds on the query complexity for testing CSPs in the bounded-degree model.

  First, for any ``symmetric'' predicate $P:\bit^{k}\to \bit$ except \equ where $k\geq 3$, 
  we show that every (randomized) algorithm that distinguishes satisfiable instances of \csp{$P$} from instances $(|P^{-1}(0)|/2^k-\epsilon)$-far from satisfiability requires $\Omega(n^{1/2+\delta})$ queries where $n$ is the number of variables and $\delta>0$ is a constant that depends on $P$ and $\epsilon$.
  This breaks a natural lower bound $\Omega(n^{1/2})$, which is obtained by the birthday paradox.
  We also show that every one-sided error tester requires $\Omega(n)$ queries for such $P$.
  These results are hereditary in the sense that the same results hold for any predicate $Q$ such that $P^{-1}(1)\subseteq Q^{-1}(1)$.
  For \equ, we give a one-sided error tester whose query complexity is $\tilde{O}(n^{1/2})$.
  Also, for \txor (or, equivalently \textsf{E2LIN2}), 
  we show an $\Omega(n^{1/2+\delta})$ lower bound for distinguishing instances between $\epsilon$-close to and $(1/2-\epsilon)$-far from satisfiability.

  Next, for the general \kcsp over the binary domain,
  we show that every algorithm that distinguishes satisfiable instances from instances $(1-2k/2^k-\epsilon)$-far from satisfiability requires $\Omega(n)$ queries.
  The matching NP-hardness is not known, even assuming the Unique Games Conjecture or the $d$-to-$1$ Conjecture.
  As a corollary, for \mis on graphs with $n$ vertices and a degree bound $d$,
  we show that every approximation algorithm within a factor $d/\poly\log d$ and an additive error of $\epsilon n$ requires $\Omega(n)$ queries.
  Previously, only super-constant lower bounds were known.
\end{abstract}
\newpage
\section{Introduction}
\textit{Property testing}~\cite{GGR98} is a relaxation of decision.
We call a randomized algorithm an \textit{$(\gamma,\epsilon)$-tester} when, 
given an oracle access $\calo_\Phi$ to an instance $\Phi$, 
it accepts $\Phi$ if it is $\gamma$-close to a predetermined property with a probability of at least $2/3$ and rejects $\Phi$ if it is $\epsilon$-far from the property with a probability of at least $2/3$.
An $(\gamma,\epsilon)$-tester is often referred to as a \textit{tolerant tester}~\cite{PRR06}.
The efficiency of an algorithm is measured by the \textit{query complexity}, which is the number of accesses to $\calo_\Phi$.
The definition of farness depends on each model.
A $(0,\epsilon)$-tester is simply called an \textit{$\epsilon$-tester}.

In this paper, we study testers for \kcsp (\textit{constraint satisfaction problems}) in the bounded-degree model and show various lower bounds on the query complexity.
An instance $\Phi$ of \kcsp is a tuple of a set of variables and a set of constraints (functions) over $k$ variables.
Then, we test whether there exists an assignment over variables that satisfies all the constraints.
We only consider Boolean CSPs.
The \textit{degree} of a variable $x$ is the number of constraints in which $x$ appears.
In the \textit{bounded-degree model}~\cite{GR08},
we only consider instances such that the degree of each variable is at most $d$, where $d$ is a predetermined parameter.
By specifying a variable $x$ and an index $i(1\leq i \leq d)$, the oracle $\calo_\Phi$ returns the $i$-th constraint in which $x$ appears.
If there exists no such constraint, $\calo_\Phi$ returns some unique symbol.
An instance $\Phi$ is called \textit{$\epsilon$-far from satisfiability} if we must remove at least $\epsilon dn/k$ constraints to make $\Phi$ satisfiable.
An instance $\Phi$ is called \textit{$\epsilon$-close to satisfiability} if we can make $\Phi$ satisfiable by removing at most $\epsilon dn/k$ constraints.
Let $P:\bit^k\to \bit$ be a predicate (a function).
Then, \csp{$P$} is a sub-problem of \kcsp in which every constraint is specified by the same predicate $P$ and literals on it (see Section~\ref{sec:preliminary} for details).
For a concrete predicate, we often use $P$ as the name of a problem instead of writing \csp{$P$} (e.g., \kxor).

The first contribution of this paper is the development of a new technique to show lower bounds for testing a wide range of \csp{$P$}.
A predicate $P$ is called \textit{symmetric} if the following conditions hold: 
(i) $P(x)=P(y)$ for any $x,y\in \bit^k$ such that $|x|=|y|$.
(ii) $P(x)=P(\overline{x})$ for any $x\in \bit^k$ where $\overline{x}=(1,\ldots,1)-x$.
We assume $|P^{-1}(1)|>0$ throughout this paper.
The simplest symmetric predicates might be \kequ$:\bit^k\to \bit$, 
which is satisfied iff the variables are all zeros or all ones,
and \knae$:\bit^k\to \bit$,
which is satisfied iff not all of the variables have the same value.
\knae is much related to coloring on $k$-uniform hypergraphs.
We show the next theorem.
\begin{theorem}\label{thr:two-sided}
  Let $P:\bit^k\to\bit$ be a symmetric predicate except \kequ where $k\geq 3$.
  Then, for any $\epsilon>0$ and predicate $Q:\bit^k\to\bit$ such that $P^{-1}(1)\subseteq Q^{-1}(1)$,
  there exist $\delta=O(1/\log (k/\epsilon^2))$ and $d=O(1/\epsilon^2)$ such that 
  every $(|Q^{-1}(0)|/2^k-\epsilon)$-tester for \csp{$Q$} with a degree bound $d$ requires $\Omega(n^{1/2+\delta})$ queries,
\end{theorem}
We note that a $(|Q^{-1}(0)|/2^k)$-tester is trivial since no instance can be $(|Q^{-1}(0)|/2^k)$-far from satisfiability and we can always accept.
Thus, Theorem~\ref{thr:two-sided} excludes the possibility of efficient \textit{non-trivial} testers.
We also stress that it is impossible to get rid of the condition of symmetry since for a certain non-symmetric CSP, 
called \textsf{Dicut}, we have a constant-time non-trivial tester using recent results~\cite{Tre98,Yos10}.
The lower bound $\Omega(n^{1/2+\delta})$ is somewhat surprising since, 
as we will see in Section~\ref{sec:overview}, 
this lower bound implies that even if we find cycles in the instance,
they do not help at all to test the satisfiability.

\kxor is a predicate of arity $k$, 
which is satisfied iff the parity of its variables is $1$.
We show a similar lower bound for \txor.
\begin{theorem}\label{thr:txor}
  For any $\epsilon>0$,
  there exist $\delta=O(\epsilon/\log(k/\epsilon^2))$ and $d=O(1/\epsilon^2)$ such that every $(\epsilon,1/2-\epsilon)$-tester for \txor with a degree bound $d$ requires $\Omega(n^{1/2+\delta})$ queries.
\end{theorem}

If an $\epsilon$-tester always accepts satisfiable instances,
it is called a \textit{one-sided error tester}.
Otherwise, it is called a \textit{two-sided error tester}.
We give a tight lower bound for one-sided error testers.
\begin{theorem}\label{thr:one-sided}
  Let $P:\bit^k\to\bit$ be a symmetric predicate except \kequ where $k\geq 3$.
  Then, for any $\epsilon>0$ and any $Q:\bit^k\to\bit$ such that $P^{-1}(1)\subseteq Q^{-1}(1)$,
  there exists $d=O(1/\epsilon^2)$ such that
  every one-sided error $(|Q^{-1}(0)|/2^k-\epsilon)$-tester for \csp{$Q$} with a degree bound $d$ requires $\Omega(n)$ queries.
\end{theorem}

On the other hand,
\kequ is an easier problem as stated in the next theorem.
\begin{theorem}\label{thr:equ}
  For any $\epsilon>0,d\geq 1$ and $k\geq 2$,
  there exists a one-sided error $\epsilon$-tester for \kequ with query complexity $O(n^{1/2}\poly(dk\log n/\epsilon))$.
\end{theorem}
\textit{Bipartiteness} is the property of a graph such that the vertex set can be partitioned into two disjoints sets $U$ and $V$ such that every edge connects a vertex of $U$ and a vertex of $V$.
Theorem~\ref{thr:equ} is almost tight since testing bipartiteness is a sub-problem of $2$-\equ and the $\Omega(\sqrt{n})$ lower bound is known for this problem~\cite{GR08}.

The second contribution of this work is a linear lower bound to distinguish satisfiable instances of the general \kcsp from instances much further from satisfiability.
\begin{theorem}\label{thr:kcsp}
  For any $\epsilon>0$ and $k\geq 3$, 
  there exists $d=O(1/\epsilon^2)$ such that 
  every $(1-2k/2^k-\epsilon)$-tester for \kcsp with a degree bound $d$ requires $\Omega(n)$ queries.
\end{theorem}
As a corollary, we show a linear lower bound for approximating \mislong (\mis).
An {\it independent set} of a graph is a vertex set such that any two of its vertices are not adjacent.
\mis is the problem of finding the largest independent set in a graph.
A value $x$ is called an \textit{$(\alpha,\beta)$-approximation} for a value $x^*$ if $x^*\leq x \leq \alpha x^*+\beta$.
We call a randomized algorithm an \textit{$(\alpha,\beta)$-approximation algorithm} for \mis if,
given an oracle access $\calo_G$ to a graph $G$,
it computes an $(\alpha,\beta)$-approximation for \mis with a probability of at least $2/3$.
Similarly to \kcsp, 
by specifying a vertex $v$ and an index $i(1\leq i \leq d)$,
the oracle $\calo_G$ returns the $i$-th edge in which $v$ appears.
We show the next theorem.
\begin{theorem}\label{thr:mis}
  Every $(d/\poly\log d,\epsilon n)$-approximation algorithm for \mislong on graphs with $n$ vertices and a degree bound $d$ requires $\Omega(n)$ queries.
\end{theorem}

\begin{table}[t]
  \caption{Summary of results on query complexity of two-sided error $(\gamma,\epsilon)$-testers for various problems. Here, $P$ denotes any symmetric predicate with arity $k\geq 3$ except \kequ.}
  \label{tbl:summary}
  \begin{center}
    \begin{tabular}{|l|l|l|l|l|}
      \hline
      Problem & $\gamma$ & $\epsilon$ & Bound & Reference \\
      \hline
      \txor & 0 & $\epsilon$ & $\tilde{\Theta}(\sqrt{n})$ & \cite{GR99,GR08}\\ 
      \hline
      \txor & $\epsilon$ & $1/12$ & $\Theta(n)$ & \cite{YI10}\\
      \hline
      \txor & $\epsilon$ & $1/2-\epsilon$ & $\Omega(n^{1/2+\delta})$ & Theorem~\ref{thr:txor} \\
      \hline
      $3$-\xor & $0$ & $1/2-\epsilon$ & $\Theta(n)$ & \cite{BOT02} \\
      \hline
      \kequ & 0 & $\epsilon$ & $\tilde{O}(\sqrt{n})$ & Theorem~\ref{thr:equ}\\
      \hline
      \kequ& 0 & $\epsilon$ & $\Omega(\sqrt{n})$ & \cite{GR08} \\
      \hline
      \csp{$P$} & 0 & $P^{-1}(0)/2^k-\epsilon$ & $\Omega(n^{1/2+\delta})$ & Theorem~\ref{thr:two-sided} \\
      \hline
      \kcsp & 0 & $1-2k/2^k-\epsilon$ & $\Theta(n)$ & Theorem~\ref{thr:kcsp} \\
      \hline
    \end{tabular}
  \end{center}
\end{table}

\paragraph{Related work:}
There have been several works on testing CSPs.
The summary of known results is shown in Table~\ref{tbl:summary}.
\maxkcsp is an optimization version of \kcsp in which we are to maximize the number of satisfied constraints by an assignment.
Let $P$ be a predicate of arity $k$.
We notice that if there is an approximation algorithm for \maxcsp{$P$} with a factor $\frac{1-\epsilon}{1-\gamma}$,
we have $(\gamma,\epsilon)$-tester for \csp{$P$}.
Thus, a lower bound for testing \csp{$P$} implies a lower bound for approximating \maxcsp{$P$}.
The NP-hardness of approximation within a certain factor is often shown using a reduction from $3$-\xor.
Using the same reduction, 
for a wide range of $P$,
it is shown that there exists some $\eta>0$ such that any $\eta$-tester requires $\Omega(n)$ queries~\cite{YI10} (e.g., $\eta=1/12$ for $\txor$~\cite{Has01} and $\eta=1/16$ for $3$-\nae~\cite{Zwi98}).
However, for $\epsilon>\eta$,
we did not have any lower bound for $\epsilon$-testers.
Theorems~\ref{thr:two-sided} and~\ref{thr:txor} tighten this gap and also imply the new lower bound $\Omega(n^{1/2+\delta})$ for approximating \maxcsp{$P$} within a factor $|P^{-1}(1)|/2^k+\epsilon$.

Assuming the Unique Games Conjecture~\cite{Kho02}, 
it is NP-hard to distinguish between instances of \maxkcsp whose optimal solutions are $1-\epsilon$ and $(k+o(k))/2^k+\epsilon$~\cite{AM08,ST06}.
Theorem~\ref{thr:kcsp} states a somewhat stronger fact about sublinear time algorithms; i.e., 
it is hard to distinguish \textit{satisfiable instances} from instances whose optimal solutions are at most $O(k)/2^k$ with sublinear queries.
No matching NP-hardness is known,
even assuming the $d$-to-$1$ Conjecture~\cite{Kho02}.

The concept of $(\alpha,\epsilon n)$-approximation algorithms was introduced in~\cite{CRT01} to approximate the minimum spanning tree of a bounded-degree graph.
Since then, 
numerous $(\alpha,\epsilon n)$-approximation algorithms have been developed for graph problems~\cite{Alo10,CRT01,MR09,NO08,PR07,YYI09}.
For \mis, it is shown that there exists a constant-time $\left(O(d \log \log d/\log d),\epsilon n\right)$-approximation algorithm,
and every $\left(o(d/\log d),\epsilon n\right)$-approximation algorithm requires a super-constant number of queries~\cite{Alo10}.
Theorem~\ref{thr:mis} improve this to a linear lower bound at the cost of a slightly weaker approximation factor.

\paragraph{Motivations:}
Among all CSPs, 
we could say that $k$-\xor is the one whose behavior is best understood.
It is NP-hard to distinguish between instances whose optimal solutions (in the sense of \maxkxor) have values $1-\epsilon$ and $1/2+\epsilon$~\cite{Has01}.
This fact means that the random assignment achieves the best approximation ratio one can obtain in polynomial time.
The behavior of \maxkxor under linear programmings (LP) and semidefinite programmings (SDP) is also well-studied.
A quality of linear and semidefinite programming are measured by \textit{integrality gap},
which is the ratio of the optimum for those programs to the optimum for the original problem.
The Lovasz-Schrijver hierarchy (LS, LS+), Sherali-Adams hierarchy (SA), Lasserre hierarchy are sequences of relaxations of those programs to obtain tighter approximations.
For all of these hierarchies, the integrality gaps remain $2-\epsilon$ after $\Omega(n)$ rounds of relaxations~\cite{BOG03,GMT09,Sch08}.

Presumably, the reason why \maxkxor is hard to approximate is that the accepting assignments of \kxor contain the support of a $(k-1)$-wise independent distribution.
These results are extended to predicates whose accepting assignments contain the support of a pairwise independent distribution~\cite{AM08, GMT09, Tul09}.
For other predicates, however, 
we can approximate better than the random assignment using SDP (e.g., $3$-\nae).
One motivation for this work is to investigate why SDP helps with those predicates.
Theorems~\ref{thr:two-sided} and~\ref{thr:txor} suggest that a few cycles are not sufficient to approximate better than the random assignment.
This holds not only for SDP,
but also for any algorithm.
Also, Theorem~\ref{thr:one-sided} gives us a separation of the ability of polynomial-time algorithms versus sublinear-time one-sided error testers since SDP approximates better than the random assignment in polynomial time.

It is an interesting question whether we can approximate \maxkcsp within a certain factor by sampling a small portion of an instance.
We can approximate the optimal solution of \maxkcsp within an additive error $\epsilon n^k$ by sampling $\poly(1/\epsilon)$ variables and by solving the induced problems~\cite{AdlVKK03}.
Thus, dense instances are easy to approximate with constant queries~\cite{AdlVKK03,AS02,AE02}.
However, little is known for sparse instances.
Solving \maxcut of a sparse graph by sampling is demonstrated in~\cite{BHHS09}.
They showed that the value of Goemans-Williamson SDP~\cite{GW95} for a randomly sampled subgraph of linear size is approximately equal to the SDP value for the original graph.
Our work is a complement of their work.
Theorem~\ref{thr:two-sided} implies that, 
to approximate symmetric \maxkcsp better than the random assignment,
we need to sample $\Omega(n^{1/2+\delta})$ constraints from the instance.

\paragraph{Organization:}
In Section~\ref{sec:preliminary}, 
we define notions used in this paper,
followed by a proof overview of Theorem~\ref{thr:two-sided}, 
which is the main result of this paper.
We give the proof of Theorem~\ref{thr:two-sided} in Section~\ref{sec:two-sided}.
We mention other results in Section~\ref{sec:other}.

\section{Preliminaries}\label{sec:preliminary}
\subsection{Definitions}
We define notions on hypergraphs.
Let $\{v_1,\ldots,v_p\}$ be a vertex set and $\{e_1,\ldots,e_{p-1}\}$ be an edge set such that $e_i$ contains $v_i$ and $v_{i+1}$ for $1 \leq i \leq p-1$.
Then, we call $\{e_1,\ldots,e_p\}$ a \textit{hyperpath}.
A hypergraph is called \textit{connected} if, for every two vertices, there is a hyperpath containing them.
Let $\{v_1,\ldots,v_p\}$ be a vertex set and $\{e_1,\ldots,e_p\}$ be an edge set such that $e_i$ contains $v_i$ and $v_{(i \bmod p)+1}$ for $1 \leq i \leq p$.
Then, we call $\{e_1,\ldots,e_p\}$ a \textit{hypercycle}.
A connected hypergraph is called a \textit{hypertree} if it does not have any hypercycle.
A \textit{hyperforest} is a hypergraph such that each connected component is a hypertree.
Let $H$ be a $k$-uniform hypergraph with $n$ vertices, $m$ edges and $c$ connected components.
We define $\cy(H)=(k-1)m-n+c$, which measures how many vertices are deficient compared to a hyperforest (note that any hyperforest with $m$ edges and $c$ connected components has $(k-1)m+c$ variables).
We call $H$ a \textit{$(\gamma,\eta)$-expander} if the subgraph of $H$ induced by any $s\leq \gamma n$ edges contains at least $(k-1-\eta)s$ vertices.

Let $P:\bit^k\to\bit$ be a predicate.
An instance $\Phi$ of \csp{$P$} is a tuple of a set of variables and a set of constraints.
Here, each constraint $C$ is defined over a $k$-tuple of variables $(x_1,\ldots,x_k)$ and is of the form $P(x_1+b_1,\ldots,x_k+b_k)=1$ for some $(b_1,\ldots,b_k)\in \bit^k$.
We call $(b_1,\ldots,b_k)$ a \textit{literal vector} of $C$.
Here, $b_i$ accounts for the possible negation of $x_i$.
The \textit{underlying hypergraph} of $\Phi$ is a $k$-uniform hypergraph $H$ in which each variable of $\Phi$ corresponds to a vertex of $H$,
and for each constraint of the form $P(x_1+b_1,\ldots,x_k+b_k)=1$ in $\Phi$, 
we have an edge $(x_1,\ldots,x_k)$ in $H$.

Next, we introduce notions on distributions.
Suppose that $\cald$ is a distribution generating $\bfx_1,\bfx_2$ (and possibly others).
Let $\cald(\bfx_1)$ be the marginal distribution of $\bfx_1$ under $\cald$.
Let $\cald(\bfx_1|\bfx_2=x_2)$ denote the marginal distribution of $\bfx_1$ conditioned on $\bfx_2=x_2$, i.e., \(\Pr_{\cald(\bfx_1|\bfx_2=x_2)}[\bfx_1]=\Pr_{\cald}[\bfx_1|\bfx_2=x_2]\).
We often omit the actual value of a random variable if it is unimportant.
For example, $\cald(\bfx_1|\bfx_2=x_2)$ may be written as $\cald(\bfx_1|\bfx_2)$ and $\sum_{x}\Pr[\bfx=x]$ may be written as $\sum_{\bfx}\Pr[\bfx]$.
Let $\supp(\cald)$ denote the support of $\cald$.
If the random variables $\bfx_1$ and $\bfx_2$ become independent after conditioning $\bfx_3$, 
we write $\bfx_1 \ci \bfx_2 \mid \bfx_3$.
Let $\{\bfx_v\}_{v\in V}$ be a set of random variables.
Then, for $S\subseteq V$, $\bfx_S$ denotes the set $\{\bfx_v\}_{v\in S}$.

Let $\cald_1$ and $\cald_2$ be distributions generating a random variable $\bfx$.
The \textit{total variation distance} between $\cald_1(\bfx)$ and $\cald_2(\bfx)$ is defined as
\[
\dtv[\cald_1(\bfx),\cald_2(\bfx)]=\sum_\bfx\left|\Pr_{\cald_1}[\bfx]-\Pr_{\cald_2}[\bfx]\right|. \footnote{This is twice as large as the standard definition. We use this definition to avoid unnecessary calculations.}
\]
We note that $0\leq \dtv[\cald_1(\bfx),\cald_2(\bfx)]\leq 2$.
Also, we define $\dtv[\cald(\bfx)] = \dtv[\cald(\bfx),\calu(\bfx)]$ where $\calu$ is the uniform distribution.
When $\bfx$ is Boolean, $0\leq \dtv[\cald(\bfx)] \leq 1$.

\subsection{Proof Overview}\label{sec:overview}
We give a proof overview of Theorem~\ref{thr:two-sided}.
To prove the lower bound, we use Yao's minimax principle~\cite{Yao77}.
Specifically, 
we design two distributions $\caldsat$ and $\caldfar$ of instances of \csp{$P$} so that all instances of $\caldsat$ are satisfiable,
while almost all instances of $\caldfar$ are $(|P^{-1}(0)|/2^k-\epsilon)$-far from satisfiability.
Then, we show that any deterministic algorithm with a sublinear number of queries cannot distinguish between instances chosen from $\caldsat$ and instances chosen from $\caldfar$.
For underlying hypergraphs of $\caldsat$ and $\caldfar$, 
we use the same distribution of expanders. 
Thus, if we ignore literal vectors and we only look at variables used in constraints,
we have no hope of distinguishing $\caldsat$ from $\caldfar$.
We describe how $\caldsat$ generates an instance.
First, $\caldsat$ chooses an underlying hypergraph $H=(V,E)$.
Then, the set of variables of the instance is $\{x_v\}_{v\in V}$.
Then, $\caldsat$ first chooses $\bfx_v\in \bit$ for each vertex $v\in V$ uniformly at random.
Here, the set $\{\bfx_v\}_{v\in V}$ is the supposed solution for the instance.
Next, $\caldsat$ chooses a literal vector $\bfb_e$ for each edge $e=(v_1,\ldots,v_k)$ and adds a constraint $C_e$ of the form $P((x_{v_1},\ldots,x_{v_k})+\bfb_e)=1$.
$\caldsat$ chooses $\bfb_e$ so that the resulting instance is satisfiable by $\{\bfx_v\}_{v\in V}$.
In contrast, $\caldfar$ simply generates $\bfb_e$ uniformly at random for each edge $e$ after choosing an underlying hypergraph.

Any algorithm with query complexity $\ell$ can be seen as a mapping from \textit{query-answer history} $(q_1,a_1),\ldots,(q_{t-1},a_{t-1})$ to $q_{t}$ for $t\leq \ell$ and to $\{\mathbf{accept},\mathbf{reject}\}$ for $t=\ell$.
A query $q_t=(v_t,i_t)$ is a pair of a variable $v_t$ and an index $i_t$, and an answer $a_t$ is a constraint or the information that there is no constraint there.
To analyze the distribution of the query-answer history of an algorithm running under a distribution of instances,
it is useful to think that there is a randomized process behind the oracle.
That is, when an algorithm asks a query of the oracle, 
the randomized process generates the answer to the query according to some distribution.
We later define a randomized process $\calpsat$ (resp., $\calpfar$),
which is equivalent to $\caldsat$ (resp., $\caldfar$) in the sense that no matter how an algorithm asks the oracle,
the distribution of instances we finally obtain is the same as $\caldsat$ (resp., $\caldfar$).
Let $\calksat$ (resp., $\calkfar$) be the distribution of query-answer history induced by the interaction between an algorithm $\cala$ and $\calpsat$ (resp., $\calpfar$).
We show that when the query complexity of $\cala$ is $o(n^{1/2+\delta})$ for some $\delta>0$,
$\dtv[\calksat,\calkfar]$ is negligibly small.
Thus, it is impossible to distinguish $\caldsat$ from $\caldfar$ with high probability.

If we ask at most $O(n^{1/2})$ queries, 
from the birthday paradox,
the query-answer history does not contain hypercycles with high probability.
From this fact, it is relatively easy to show that we cannot distinguish $\caldsat$ from $\caldfar$ with $O(n^{1/2})$ queries.
However, if we ask $\Omega(n^{1/2+\delta})$ queries,
the situation completely changes because of the effect of hypercycles.
For example, suppose that the predicate is \equ and $\cala$ obtained a constraint $C_e$ such that variables $x_u,x_v\in C_e$ already appeared in the query-answer history.
Then, $\cala$ can calculate the parity $x_u\oplus x_v$, by the propagation, along the constraint $C_e$ and along a path in the query-answer history.
If they are not the same, the instance must come from $\caldfar$.
In other words, if we assume that the instance comes from $\caldsat$,  
we can guess $\bfb_e$ from the query-answer history.

Can we generalize this algorithm to other predicates?
Though we do not exclude the possibility of sublinear-time algorithms,
we can show that, in general, we need quite a few hypercycles to distinguish $\caldsat$ from $\caldfar$.
The reason why we were able to use the propagation is that the value of a variable in a predicate \equ uniquely determines the values of other variables.
For other symmetric predicates, however, this is not true.
In fact, the correlation between variables exponentially decays along paths.
Thus, even if variables $x_u$ and $x_v$ already appeared in the query-answer history,
the correlation between $\bfx_u$ and $\bfx_v$ is tiny (before obtaining $C_e$).
Precisely, 
we will show that $\dtv[\caldsat(\bfx_u|\bfx_v,\bfb_{E'}),\caldsat(\bfx_u|\bfb_{E'})]$ is tiny where $E'$ is the edge set in the query-answer history.
Thus, $\bfb_e$ is almost identical to the uniform distribution.
It follows that we cannot distinguish $\caldsat$ from $\caldfar$ with $O(n^{1/2+\delta})$ queries.

To prove this, we use several facts about expanders.
Note that the lengths of hypercycles are large (roughly, $g=\Theta(\log_d n)$) in an expander.
Thus, for two adjacent vertices $u$ and $v$,
the distance between them is at least $g$ after removing the constraint containing them.
Furthermore, the neighborhood of $v$ looks like a hypertree $T$ with depth $g$.
Note that any information from $\bfx_u$ comes through the leaves of $T$.
Though the number of leaves of $T$ is exponential in the depth, 
we can show that the only tiny portion of them is connected to $u$ (without passing $v$).
Since such leaves have an exponentially small correlation with $\bfx_v$,
we conclude that the correlation between $\bfx_u$ and $\bfx_v$ is negligibly small.

\subsection{Properties of $\dtv$}\label{subsec:probability}
We show several lemmas about $\dtv$ and probability distributions.
Due to the space limit, all the proofs are deferred to Appendix~\ref{apx:probability}.
\begin{lemma}\label{lmm:addition}
  Let $\cald_1$ and $\cald_2$ be distributions generating random variables $\bfx$ and $\bfy$.
  Suppose that \( \dtv[\cald_1(\bfx),\cald_2(\bfx)]\leq \delta_x\), and \( \dtv[\cald_1(\bfy|\bfx=x),\cald_2(\bfy|\bfx=x)]\leq \delta_y\) for any $x$.
  Then, \(\dtv[\cald_1(\bfx,\bfy),\cald_2(\bfx,\bfy)]\leq \delta_x + \delta_y \).
\end{lemma}
\begin{lemma}\label{lmm:serial}
  Let $\cald$ a distribution generating $\bfx_A,\bfx_B,\bfx_C$.
  Suppose that $\bfx_A\ci\bfx_C\mid\bfx_B$.
  Then,
  \[
  \dtv[\cald(\bfx_C|\bfx_A)]\leq \dtv[\cald(\bfx_B|\bfx_A)]\cdot \dtv[\cald(\bfx_C|\bfx_B)].
  \]
\end{lemma}
\begin{lemma}\label{lmm:product}
  Let $\cald$ be a distribution generating $\bfx,\bfy_i(1\leq i \leq k)$.
  Suppose that $\Pr_{\cald}[\bfx=x]$ is equal for every $x\in \supp(\cald(\bfx))$ and $\bfy_i\ci\bfy_j \mid \bfx$ for every $1\leq i,j\leq k$.
  Then,
  \begin{eqnarray*}
    \Pr_{\cald}[\bfx=x|\{\bfy_i\}_{i=1}^k]=\frac{\prod_{i=1}^{k}\Pr_{\cald}[\bfx=x|\bfy_i]}{\sum_{x'\in \supp(\cald(\bfx))}\prod_{i=1}^{k}\Pr_{\cald}[\bfx=x'|\bfy_i]}.
  \end{eqnarray*}
\end{lemma}

\section{An $\Omega(n^{1/2+\delta})$ Lower Bound for Two-Sided Error Testers}\label{sec:two-sided}
In this section, 
we give a proof of Theorem~\ref{thr:two-sided}.
A reader can safely assume that a predicate $P$ is symmetric until the proof of Theorem~\ref{thr:two-sided}.

\subsection{Probabilistic Constructions of Expanders}\label{sec:expander}
We introduce a probability distribution $\calg_{n,d,k}$ of $d$-regular $k$-uniform multi-hypergraphs with $n$ vertices.
This distribution is used to define $\caldsat$ and $\caldfar$.
Here, we assume that $dn$ is divisible by $k$ (otherwise, no $d$-regular $k$-uniform hypergraph exists).
We construct a hypergraph $H=(V,E)$ as follows.
We start with a set of $dn$ vertices $V'$ where a vertex $v\in V$ is corresponding to $d$ vertices in $V'$.
Then, we partition $V'$ into $k$-hyperedges randomly.
Finally, we contract each $d$ vertices of $V'$ and let $H$ be the resulting graph.
The proof of the following lemma is deferred to Appendix~\ref{apx:expander}.
\begin{lemma}\label{lmm:expander}
  Let $H$ be a hypergraph chosen uniformly at random from $\calg_{n,d,k}$.
  For any $\eta$,
  there exists $\gamma$ such that $H$ is a $(\gamma,\eta)$-expander with probability $1-o(1)$.
\end{lemma}

\subsection{Hard instances}\label{subsec:hard-instance}
As in the proof overview, 
we introduce two distributions $\caldsat$ and $\caldfar$ of instances of \csp{$P$}.
First, we define a distribution generating instances of \csp{$P$} given an underlying hypergraph.
\begin{definition}
  Let $H=(V,E)$ be a $k$-uniform hypergraph with $n$ vertices.
  Define a distribution $\cald_H$ generating an instance $\Phi$ of \csp{$P$} as follows.
  The variable set of $\Phi$ is $\{x_v\}_{v\in V}$.
  We choose $\bfx\in \bit^n$ uniformly at random.
  For each edge $e=(v_1,\ldots,v_k)\in E$, 
  we choose $\bfb_e$ uniformly at random from the set $\{b\in \bit^k\mid P((\bfx_{v_1},\ldots,\bfx_{v_k})+b)=1 \}$.
  Then, we add a constraint $C_e$ of the form $P((x_{v_1},\ldots,x_{v_k})+\bfb_e)=1$ to $\Phi$.
\end{definition}
\begin{definition}
  Given parameters $n,d,k$, 
  define a distribution $\caldsat$ generating an instance of \csp{$P$} as follows.
  First, we choose a hypergraph $H$ from $\calg_{n,d,k}$.
  Then, an instance is output according to $\cald_H$.

  Similarly, define a distribution $\caldfar$ generating an instance of \csp{$P$} as follows.
  First, we choose a hypergraph $H=(V,E)$ from $\calg_{n,d,k}$.
  Then, for each edge $e=(v_1,\ldots,v_k)\in E$,
  we choose $b\in \bit^k$ uniformly at random and add a constraint $C_e$ of the form $P((x_{v_1},\ldots,x_{v_k})+b)=1$.
\end{definition}
We can describe the generating process of $\caldsat$ with a graphical model.
Each vertex in the graphical model corresponds to $\bfx_v (v\in V)$ or $\bfb_e (e\in E)$,
and each edge expresses the dependency between two random variables.
For an exposition of graphical models, see~\cite{Bis06}.
The important fact derived from the graphical model is the following.
\begin{observation}
  Let $H$ be a hypergraph and $G=(V,E)$ be a subgraph of $H$.
  Let $A,B,C$ be sets of vertices such that any path in $G$ between $A$ and $C$ passes a vertex of $B$.
  Then, $\bfx_A\ci \bfx_C \mid \bfx_B$ under $\cald_H(\cdot|\bfb_E)$.
\end{observation}

From the construction, 
any instance of $\caldsat$ is satisfiable.
On the other hand, 
the following lemma is well-known (e.g.,~\cite{Sch08,Tul09}).
We provide a proof for completeness in Appendix~\ref{apx:far}.
\begin{lemma}\label{lmm:far}
  For any $\epsilon>0$, there exists an integer $d\geq 1$ for which the following holds.
  Let $\Phi$ be an instance of \csp{$P$} chosen from $\caldfar$ where $P:\bit^k\to\bit$ is a predicate.
  Then, $\Phi$ is $(|P^{-1}(0)|/2^k-\epsilon)$-far from satisfiability with a probability of $1-o(1)$.
\end{lemma}

\subsection{Randomized processes equivalent to $\caldsat$ and $\caldfar$}
We show that, with high probability,
any algorithm $\cala$ with $O(n^{1/2+\delta})$ queries runs on distributions $\caldsat$, or $\caldfar$ can find at most $O(n^{3\delta})$ cycles and the lengths of those cycles are $\Omega(\log_{dk}n)$.

We define a randomized process $\calpsat$, which interacts with $\cala$, 
so that $\calpsat$ answers queries from $\cala$ while constructing a random graph from $\caldsat$.
Thus, the interaction of $\calpsat$ with $\cala$ captures a random execution of $\cala$ on a graph uniformly distributed in $\caldsat$.
Similarly, we define a randomized process $\calpfar$, which imitates $\caldfar$.

The process $\calpsat$ has two stages.
The first stage continues as long as $\cala$ performs queries, and $\calpsat$ answers to those queries.
In the second stage, $\calpsat$ determines the rest of the instance.
$\calpsat$ internally holds a supposed solution $\{\bfx_v\}_{v\in V}$, 
which is hidden from $\cala$.
Literal vectors are determined so as not to contradict this solution.

First stage of $\calpsat$:
Starting from $t=1$, 
for each query $q_t=(v_t,i_t)$ of $\cala$, $\calpsat$ proceeds as follows.
For each vertex $u$, 
we define \textit{remaining degree} $r(u)$ as the number of constraints adjacent to $u$ which are not accessed yet by $\cala$.
We choose $u_2,\ldots,u_k$ with a probability according to their remaining degrees.
Specifically, 
since the sum of remaining degrees of all vertices at the time that $\cala$ specifies $v_t$ is $dn-(t-1)k+1$,
the probability that a vertex $u$ is chosen as $u_2$ is \( r(u)/(dn-(t-1)k-1) \).
Similarly, the probability that $u$ is chosen as $u_3$ is \( r(u)/(dn-(t-1)k-2) \) since the sum of the remaining degrees decreases by one.
This process continues until $u_k$ is chosen.
Finally, form an edge $e=(u_1=v_t,\ldots,u_k)$.
For each chosen vertex $u_i\in e$, 
if the supposed solution $\bfx_{u_i}$ is not determined yet,
$\calpsat$ chooses $\bfx_{u_i}\in \bit$ uniformly at random.
Then, $\calpsat$ chooses a literal vector $\bfb_e\in \bit^k$ uniformly at random from \( \{ b\in \bit^k \mid P((\bfx_{u_1},\ldots,\bfx_{u_t})+b)=1\}  \).
Finally, $\calpsat$ returns the constraint $C_e$ of the form $P((x_{u_1},\ldots,x_{u_t})+\bfb_e)=1$ to $\cala$.

Second stage of $\calpsat$:
Among all possibilities of the rest of the underlying graph,
$\calpsat$ chooses one of them uniformly at random.
Then, $\calpsat$ decides $\bfx_v$ and $\bfb_e$ randomly in the same way as the first stage.

The process $\calpfar$ proceeds in an almost identical manner.
The only difference is that $\calpfar$ does not keep track of the supposed solution and always chooses literal vectors uniformly at random.
It is easy to confirm that the following lemma holds using indunction on the number of queries,
and we omit the proof (see Lemma~7.3 of~\cite{GR08} for details).
\begin{lemma}\label{lmm:equivalent}
  For every algorithm $\cala$, the process $\calpsat$ (resp., $\calpfar$) uniformly generates instances of $\caldsat$ (resp., $\caldfar$) when interacting with $\cala$.
  \qed
\end{lemma}
The proofs of the following two lemmas are deferred to Appendices~\ref{apx:few-cycles} and~\ref{apx:large-girth}.
\begin{lemma}\label{lmm:few-cycles}
  Let $\delta\geq 0$ and $G$ be the hypergraph induced by the query-answer history after $O(n^{1/2+\delta})$ steps of interactions between an algorithm $\cala$ and $\calpsat$ (or $\calpfar$).
  Then, with a probability of at least $1-o(1)$, $\cy(G)=O(k^2n^{3\delta})$.
\end{lemma}
\begin{lemma}\label{lmm:large-girth}
  Let $\delta \geq 0$ and $G$ be the hypergraph induced by the query-answer history after $O(n^{1/2+\delta})$ steps of interactions between an algorithm $\cala$ and $\calpsat$ (or $\calpfar$).
  Then, with a probability of at least $1-o(1)$, 
  the girth of $G$ is at least $g=(\frac{1}{2}-2\delta)\log_{dk}n$.
\end{lemma}

\subsection{Correlation decay along edges of a hypertree}
Let $\Phi$ be an instance of \csp{$P$} generated by $\caldsat$.
Suppose that $T=(V,E)$ is a subgraph of the underlying graph of $\Phi$ and $T$ is a hypertree.
Let $v\in V$ be a (arbitrary) root of $T$ and $L$ be a subset of leaves of $T$.
In this subsection,
we consider how the information of $\bfx_L$ propagates into $\bfx_v$ along edges of $T$.
Specifically, we calculate $\dtv[\cald_H(\bfx_v|\bfx_L,\bfb_E),\cald_H(\bfx_v|\bfb_E)]$.
A proof of the next lemma is given in Appendix~\ref{apx:cannot-guess-tree}.
\begin{lemma}\label{lmm:cannot-guess-tree}
  Let $T=(V,E)$ be a subgraph of a hypergraph $H$.
  If $T$ is a hypertree, 
  then $\bfx_v$ and $\bfb_E$ are independent for any $v\in V$ under $\cald_H$.
\end{lemma}
From Lemma~\ref{lmm:cannot-guess-tree},
$\dtv[\cald_H(\bfx_v|\bfx_L,\bfb_E),\cald_H(\bfx_v|\bfb_E)]=\dtv[\cald_H(\bfx_v|\bfx_L,\bfb_E)]$ holds.

Next, we see how $\dtv[\cald_H(\bfx_v|\bfx_L,\bfb_E)]$ propagates by connecting vertices at a vertex or an edge.
Proofs of next two lemmas are given in Appendices~\ref{apx:merge-vertex}, and~\ref{apx:merge-edge}, respectively.
\begin{lemma}\label{lmm:merge-vertex}
  Let $T=(V,E)$ be a subgraph of a hypergraph $H$.
  Suppose that $T$ is a hypertree.
  Let $T_1,\ldots,T_\ell$ be the set of the subtrees obtained by splitting $v\in V$ and $L_i (1\leq i\leq \ell)$ be a subset of the leaves of $T_i$.
  Then,
  \begin{eqnarray*}
    \dtv[\cald_H(\bfx_v|\{\bfx_{L_i}\}_{i=1}^\ell,\bfb_E)] \leq \sum_{i=1}^{\ell}\dtv[\cald_H(\bfx_v|\bfx_{L_i},\bfb_E)].
  \end{eqnarray*}
\end{lemma}
\begin{lemma}\label{lmm:merge-edge}
  Let $T=(V,E)$ be a subgraph of a hypergraph $H$.
  Suppose that $T$ is a hypertree.
  Let $T_1,\ldots,T_{k-1}$ be the set of subtrees obtained by removing $e=(v_1,\ldots,v_k)\in E$.
  Here, $v_i$ is the root of $T_i$.
  Let $L_i(1\leq i\leq k-1)$ be a subset of the leaves of $T_i$.
  Then,
  \begin{eqnarray*}
    \dtv[\cald_H(\bfx_{v_k}|\{\bfx_{L_i}\}_{i=1}^{k-1},\bfb_E)]\leq \rho(P)\sum_{i=1}^{k-1}\dtv[\cald_H(\bfx_{v_i}|\bfx_{L_i},\bfb_E)].
  \end{eqnarray*}
  Here, $\rho(P)\leq 1$ is a constant, which only depends on the (symmetric) predicate $P$.
  In particular, $\rho(P)<1$ if $P$ is not \equ.
\end{lemma}

\subsection{Putting things together}
Let $\rho(P)$ be the constant determined in Lemma~\ref{lmm:merge-edge}.
\begin{lemma}\label{lmm:E+S-to-v}
  Let $G=(V,E)$ be a subgraph of a hypergraph $H$ with girth $g$ and let $e\in E$ be an edge.
  Then, for any $v\in e$ and $S\subseteq e-\{v\}$, $\dtv[\cald_H(\bfx_v|\bfx_S,\bfb_{E-e})]\leq \rho(P)^{g}(2\cy(G-e)+k)$.
\end{lemma}
\begin{proof}
  Let $T=(V_T,E_T)$ be a subgraph of $G$ induced by vertices whose distance from $v$ in $G-e$ is at most $g$.
  Note that $T$ is a hypertree rooted at $v$ since the girth of $G-e$ is $g$.
  For a leaf $u$ of $T$, let $C_u$ be the resulting connecting component containing $u$ after removing $T$.
  We define $L$ as a subset of leaves as follows.
  A leaf $u$ is in $L$ iff $C_u$ contains a vertex of $S$ or $C_u$ is not a hypertree.
  Once $L$ is connected to all vertices of $S$, 
  each leaf $u\in L$ involves a cycle.
  Thus, $|L|$ is at most $2\cy(G-e)+k$.
  From Lemma~\ref{lmm:serial}, 
  \begin{eqnarray*}
    \dtv[\cald_H(\bfx_v|\bfx_S,\bfb_{E-e})]\leq \dtv[\cald_H(\bfx_v|\bfx_L,\bfb_{E-e})]\dtv[\cald_H(\bfx_L|\bfx_S,\bfb_{E-e})] \leq 2\dtv[\cald_H(\bfx_v|\bfx_L,\bfb_{E-e})].
  \end{eqnarray*}
  For each leaf $u\neq L$ of $T$, 
  we can truncate edges of $C_u$ since they have no information about $\bfx_u$ from Lemma~\ref{lmm:cannot-guess-tree}.
  Also, $\bfb_{E_T} \ci \bfb_{E-e-E_T} \mid \bfx_L$.
  Thus, \( \dtv[\cald_H(\bfx_v|\bfx_S,\bfb_{E-e})] \leq 2\dtv[\cald_H(\bfx_v|\bfx_L,\bfb_{E_T})] \).
  Now, to calculate $\dtv[\cald_H(\bfx_v|\bfx_{L},\bfb_{E_T})]$, 
  we recursively use Lemmas~\ref{lmm:merge-vertex} and~\ref{lmm:merge-edge} from leaves.
  For each leaf $u$ of $T$, 
  we consider $\dtv[\cald_H(\bfx_u|\bfx_{L_u},\bfb_{E_T})]$ where $L_u=\{u\}\cap L$.
  We note that \(\dtv[\cald_H(\bfx_u|\bfx_{L_u},\bfb_{E_T})]=1\) for $u\in L$,
  and \(\dtv[\cald_H(\bfx_u|\bfx_{L_u},\bfb_{E_T})]=0\) for a leaf $u\not \in L$.
  Then, it is clear that \( \dtv[\cald_H(\bfx_v|\bfx_{L},\bfb_{E_T})] \leq \rho(P)^g|L|\leq \rho(P)^g(2\cy(G-e)+k)\).
\end{proof}

\begin{lemma}\label{lmm:E-to-e}
  Let $G=(V,E)$ be a subgraph of a hypergraph $H$ with girth $g$ and let $e\in E$ be an edge in a hypercycle of $G$.
  Then, \(\dtv[\cald_H(\bfb_e|\bfb_{E-e})]\leq k \rho(P)^g(4\mathrm{cy}(G-e)+2k) \).
\end{lemma}
\begin{proof}
  From Lemma~\ref{lmm:serial},
  \(\dtv(\cald_H[\bfb_e|\bfb_{E-e}]) \leq \dtv[\cald_H(\bfx_e|\bfb_{E-e})]\dtv[\cald_H(\bfb_e|\bfx_e)] \leq 2\dtv[\cald_H(\bfx_e|\bfb_{E-e})]\).
  Let $e=(v_1,\ldots,v_k)$.
  From Lemma~\ref{lmm:addition}, 
  \(\dtv[\cald_H(\bfx_e|\bfb_{E-e})] \leq \sum_{i=1}^k\dtv[\cald_H(\bfx_{v_i}|\{\bfx_{v_{j}}\}_{j=1}^{i-1},\bfb_{E-e})] \).
  From Lemma~\ref{lmm:E+S-to-v}, 
  we have \( \dtv[\cald_H(\bfx_{v_i}|\{\bfx_{v_{j}}\}_{j=1}^{i-1},\bfb_{E-e})]\leq \rho(P)^g(2\mathrm{cy}(G-e)+k) \) for $2\leq i\leq k$.
  This inequality holds for $i=1$ since $\cald_H(\bfx_{v_1}|\bfb_{E-e})$ is a convex combination of $\cald_H(\bfx_{v_1}|\bfx_{v_2}=0,\bfb_{E-e})$ and $\cald_H(\bfx_{v_1}|\bfx_{v_2}=1,\bfb_{E-e})$.
  Thus, the lemma holds.
\end{proof}

We show a weaker version of Theorem~\ref{thr:two-sided}, which is only for symmetric predicates.
\begin{theorem}\label{thr:two-sided-weak}
  Let $P:\bit^k\to\bit$ be a symmetric predicate except \kequ where $k\geq 3$.
  Then, for any $\epsilon>0$,
  there exist $\delta>0$ and $d\geq 1$ such that 
  every $(|P^{-1}(0)|/2^k-\epsilon)$-tester for \csp{$P$} with a degree bound $d$ requires $\Omega(n^{1/2+\delta})$ queries,
  where $\delta=O(1/\log (k/\epsilon^2))$.
\end{theorem}
\begin{proof}
  Suppose that there exists a deterministic $(|P^{-1}(0)|/2^k-\epsilon)$-tester $\cala$ for \csp{$P$} with query complexity $t=o(n^{1/2+\delta})$.
  We choose $\delta>0$ later.
  From Lemmas~\ref{lmm:few-cycles} and~\ref{lmm:large-girth},
  by union bound,
  $\cala$ finds vertices in the current query-answer history $c=O(k^2n^{3\delta})$ times at most and the length of found cycles is at least $g=(1/2-2\delta)\log_{dk}n$ with a probability $1-o(1)$.
  In what follows, we condition on these events.

  We consider a decision tree $\caltsat$ generated by interactions between $\cala$ and $\calpsat$.
  To define $\caltsat$, 
  we suppose that an interaction between $\cala$ and $\calpsat$ proceeds in two steps;
  i.e., $\calpsat$ first returns a set of $k$ variables $X$ that will be used in the answer constraint,
  and then returns a literal vector $b$ for the constraint.
  Corresponding to these two steps,
  $\caltsat$ has two kinds of vertices, i.e., \textit{S-vertices (state vertices)} and \textit{I-vertices (intermediate vertices)}.
  In any path of the tree from the root to a leaf,
  $S$-vertices and $I$-vertices appear alternately.
  Each $S$-vertex $v$ corresponds to a particular state of the query-answer history.
  When $\cala$ obtains a set of variables $X$ from $\calpsat$,
  the state proceeds to a $I$-vertex $u$,
  which is a child of $v$.
  The edge $(v,u)$ is associated with $X$ and the transition probability.
  After that, $\cala$ obtains a literal vector $b$ from $\calpsat$.
  Then, the state proceeds to an $S$-vertex $v'$,
  which is a child of $u$.
  The edge $(u,v')$ is associated with $b$ and the transition probability.
  The tree $\caltfar$ is similarly defined.
  In particular,
  $\caltsat$ and $\caltfar$ are isomorphic.
    
  We consider couplings of corresponding vertices in $\caltsat$ and $\caltfar$ (one is mapped to another by the isomorphism).
  Suppose that $v_{\mathrm{sat}}$ and $v_{\mathrm{far}}$ are a pair of coupled $S$-vertices.
  Since the distributions of variables returned by the first step of an interaction are identical between $\calpsat$ and $\calpfar$,
  the transition distributions to their children are identical.
  Next, suppose that $v_{\mathrm{sat}}$ and $v_{\mathrm{far}}$ is a pair of coupled $I$-vertices.
  If the constraint returned by the previous step does not form a new hypercycle,
  the transition distributions to their children are identical.
  If the constraint forms a new hypercycle,
  from Lemma~\ref{lmm:E-to-e},
  the total variation distance between the distributions of literal vectors is at most $k(4c+2k)\rho(P)^g$.
  With the probability corresponding to this distance,
  we suppose that $\cala$ succeeds in distinguishing $\calpsat$ from $\calpfar$ and terminates.
  This always makes $\cala$ more powerful.
  After this modification, 
  the transition distributions to their children become identical.
  After all, for any pair of coupled leaves of $\caltsat$ and $\caltfar$, 
  the transition probabilities from the root to them are the same.
  Thus, we cannot distinguish $\calpsat$ from $\calpfar$ if we reach a leaf of the decision tree.

  Thus, it amounts to calculate the sum of the discarded probabilities.
  Suppose any path $p$ from the root to a leaf of $\caltsat$.
  Then, since $\cala$ finds vertices in the query-answer history $c$ times at most,
  the sum of the discarded probability in $p$ is at most $ck(4c+2k)\rho(P)^g$.
  Since the tree is a convex combination of paths (with respect to transition probabilities),
  the total discarded probability is at most $ck(4c+2k)\rho(P)^g$.
  By choosing $\delta=O(1/\log (k/\epsilon^2))$, this value becomes a small constant.
  Thus, $\dtv[\calksat,\calkfar] \ll 1$ where $\calksat$ (resp., $\calkfar$) is the distribution of the query-answer history induced by the $t$ steps of interactions between $\cala$ and $\calpsat$ (resp., $\calpfar$).

  Since $\cala$ is a $(|P^{-1}(0)|/2^k-\epsilon)$-tester, 
  $\Pr[\cala(\calksat)=\mathrm{accept}]\geq 2/3$.
  On the other hand, 
  since a $1-o(1)$ fraction of instances of $\caldfar$ is $(|P^{-1}(0)|/2^k-\epsilon)$-far from satisfiability,
  $\Pr[\cala(\calkfar)=\mathrm{accept}]\leq (1-o(1))\frac{1}{3}+o(1)=\frac{1}{3}+o(1)$.
  This is, however, a contradiction since $\dtv[\calksat,\calkfar] \ll 1 $.
\end{proof}

\begin{proof}[Proof of Theorem~\ref{thr:two-sided}]
Let $Q:\bit^k\to\bit$ be a predicate such that $P^{-1}(1)\subseteq Q^{-1}(1)$ for a symmetric predicate $P:\bit^k\to\bit$ except \kequ.
We slightly change the definition of $\cald_H$.
That is, for each edge $e=(v_1,\ldots,v_k)$ of $H$,
we choose $\bfb_e$ uniformly at random from the set $\{b\in \bit^k\mid P((\bfx_{v_1},\ldots,\bfx_{v_k})+b)=1 \}$ instead of $\{b\in \bit^k\mid Q((\bfx_{v_1},\ldots,\bfx_{v_k})+b)=1 \}$.
Then, the rest of the proof is the same as the proof of Theorem~\ref{thr:two-sided-weak}.
\end{proof}

\section{Other Results}\label{sec:other}
A proof of Theorem~\ref{thr:txor} is given in Appendix~\ref{apx:txor}.
The proof is similar to the proof of Theorem~\ref{thr:two-sided}.
The modifications we need are the construction of $\caldsat$.
Specifically, 
we introduce noise to the literal vector of each constraint with probability $\epsilon$ so that we cannot use propagation anymore to guess the value of variables.
A proof of Theorem~\ref{thr:one-sided} is given in Appendix~\ref{apx:one-sided}.
Any one-sided error tester $\cala$ cannot reject an instance $\Phi$ until $\cala$ finds the evidence that $\Phi$ is not satisfiable.
For a hard instance, 
we use an instance obtained from $\caldfar$, which is defined in Section~\ref{subsec:hard-instance}.
We show that it is far from satisfiability while any linear size sub-instance of it is satisfiable.
This leads to a linear lower bound.
A proof of Theorem~\ref{thr:equ} is given in Appendix~\ref{apx:equ}.
We reduce \kequ to the problem of testing bipartiteness of a graph,
and finally we use a tester for bipartiteness given in~\cite{GR99}.
Proofs of Theorems~\ref{thr:kcsp} and~\ref{thr:mis} are given in Appendix~\ref{apx:kcsp}.
Similar to~\cite{Tul09}, 
we define a predicate $P$ using Hamming code.
Using algebraic properties of Hamming code, 
we show that the \csp{$P$} is hard to test with sublinear queries.
Our proof can be seen as an extension of the proof of~\cite{BOT02}, 
which showed a linear lower bound for testing $3$-\xor.
We prove Theorem~\ref{thr:mis} using a reduction from the hardness of \kcsp.

\newpage

\section*{Acknowledgements}
The author thanks Daisuke Okanohara and Masaki Yamamoto for helpful comments.

\bibliography{arxiv}{}

\newpage
\appendix
\noindent {\bf\Large \appendixname}

\section{Proof of Subsection~\ref{subsec:probability}}\label{apx:probability}

\subsection{Proof of Lemma~\ref{lmm:addition}}
\begin{proof}
  \begin{eqnarray*}
    &&
    \sum_{\bfx,\bfy}\left|\Pr_{\cald_1}[\bfx,\bfy]-\Pr_{\cald_2}[\bfx,\bfy]\right|
    =
    \sum_{\bfx,\bfy}\left|\Pr_{\cald_1}[\bfx]\Pr_{\cald_1}[\bfy|\bfx]-\Pr_{\cald_2}[\bfx]\Pr_{\cald_2}[\bfy|\bfx]\right|\\
    &=&
    \sum_{\bfx,\bfy}\left|\Pr_{\cald_1}[\bfx]-\Pr_{\cald_2}[\bfx]\right|\Pr_{\cald_1}[\bfy|\bfx]
    +\sum_{\bfx,\bfy}\Pr_{\cald_2}[\bfx]\left|\Pr_{\cald_1}[\bfy|\bfx]-\Pr_{\cald_2}[\bfy|\bfx]\right|\\
    &\leq&
    \sum_{\bfy}\delta_x \Pr_{\cald_1}[\bfy|\bfx] + \sum_{\bfx}\Pr_{\cald_2}[\bfx]\delta_y
    =
    \delta_x+\delta_y.
  \end{eqnarray*}
\end{proof}

\subsection{Proof of Lemma~\ref{lmm:serial}}
\begin{proof}
  We consider the following value.
  \begin{eqnarray*}
    &&
    \sum_{\bfx_B}(\Pr[\bfx_B|\bfx_A]-\Pr[\bfx_B])(\Pr[\bfx_C|\bfx_B]-\Pr[\bfx_C])\\
    &=&
    \sum_{\bfx_B}\left((\Pr[\bfx_B|\bfx_A]-\Pr[\bfx_B])\Pr[\bfx_C|\bfx_B]-(\Pr[\bfx_B|\bfx_A]-\Pr[\bfx_B])\Pr[\bfx_C]\right)\\
    &=&
    \sum_{\bfx_B}\left(\Pr[\bfx_C,\bfx_B|\bfx_A]-\Pr[\bfx_B,\bfx_C]\right)\;\;\;(\mbox{from }\bfx_A\ci\bfx_C\mid\bfx_B)\\
    &=&
    \Pr[\bfx_C|\bfx_A]-\Pr[\bfx_C].
  \end{eqnarray*}
  Then,
  \begin{eqnarray*}
    \dtv[\cald(\bfx_C|\bfx_A),\cald(\bfx_C)]
    &=&
    \sum_{\bfx_C}\left|\Pr[\bfx_C|\bfx_A]-\Pr[\bfx_C]\right|\\
    &=&
    \sum_{\bfx_C}\left|\sum_{\bfx_B}(\Pr[\bfx_B|\bfx_A]-\Pr[\bfx_B])(\Pr[\bfx_C|\bfx_B]-\Pr[\bfx_C])\right|\\
    &\leq&
    \sum_{\bfx_C}\sum_{\bfx_B}\left(|\Pr[\bfx_B|\bfx_A]-\Pr[\bfx_B]|\cdot |\Pr[\bfx_C|\bfx_B]-\Pr[\bfx_C]|\right)\\
    &\leq&
    \sum_{\bfx_B}|\Pr[\bfx_B|\bfx_A]-\Pr[\bfx_B]|\cdot \sum_{\bfx_C}|\Pr[\bfx_C|\bfx_B]-\Pr[\bfx_C]|\\
    &\leq&
    \dtv[\cald(\bfx_B|\bfx_A)]\cdot \dtv[\cald(\bfx_C|\bfx_B)].
  \end{eqnarray*}
\end{proof}

\subsection{Proof of Lemma~\ref{lmm:product}}
\begin{proof}
  We use induction on $k$.
  When $k=1$, we have nothing to prove.

  Suppose that the lemma holds when $k<t$.
  We will show that the lemma also holds when $k=t$.
  In fact,
  \begin{eqnarray*}
    &&
    \Pr[\bfx=x|\{\bfy_i\}_{i=1}^t]\\
    &=&
    \frac{\Pr[\bfx=x|\bfy_t]\Pr[\{\bfy_i\}_{i=1}^{t-1}|\bfx=x,\bfy_t]}{\Pr[\{\bfy_i\}_{i=1}^{t-1}|\bfy_t]}\\
    &=&
    \frac{\Pr[\bfx=x|\bfy_t]\Pr[\{\bfy_i\}_{i=1}^{t-1}|\bfx=x]}{\sum_{x'\in \supp(\cald(x))}\Pr[\bfx=x'|\bfy_t]\Pr[\{\bfy_i\}_{i=1}^{t-1}|\bfx=x']}\;\;\;(\mbox{from }\bfy_t \ci \bfy_i \mid \bfx)\\
    &=&
    \frac{\Pr[\bfx=x|\bfy_t]\Pr[\bfx=x|\{\bfy_i\}_{i=1}^{t-1}]\Pr[\{\bfy_i\}_{i=1}^{t-1}]/\Pr[\bfx=x]}{\sum_{x'\in \supp(\cald(x))}\Pr[\bfx=x'|\bfy_t]\Pr[\bfx=x'|\{\bfy_i\}_{i=1}^{t-1}]\Pr[\{\bfy_i\}_{i=1}^{t-1}]/\Pr[\bfx=x']}\\
    &=&
    \frac{\Pr[\bfx=x|\bfy_t]\Pr[\bfx=x|\{\bfy_i\}_{i=1}^{t-1}]}{\sum_{x'\in \supp(\cald(x))}\Pr[\bfx=x'|\bfy_t]\Pr[\bfx=x'|\{\bfy_i\}_{i=1}^{t-1}]}.\;\;\;(\Pr[\bfx=x]\mbox{ is equal for every }x)
  \end{eqnarray*}
  Substituting \(\Pr[\bfx=x|\{\bfy_i\}_{i=1}^{t-1}]=\frac{\prod_{i=1}^{t-1}\Pr_{\cald}[\bfx=x|\bfy_i]}{\sum_{x'\in \supp(\cald(x))}\prod_{i=1}^{t-1}\Pr_{\cald}[\bfx=x'|\bfy_i]}\),
  we get the desired result.
\end{proof}

\section{Proof of Section~\ref{sec:two-sided}}\label{apx:proof-two-sided}

\subsection{Proof of Lemma~\ref{lmm:expander}}\label{apx:expander}
\begin{proof}
  Fix a set of $s$ (random) hyperedges $S$ and a set of $cs$ vertices $X$ where $c=k-1-\eta$.
  We consider the probability that every hyperedge of $S$ is contained in $X$.
  Since $cs$ vertices are involved with at most $csd$ hyperedges, $s$ hyperedges determine at most $ks$ neighbors, 
  This is upper-bounded by 
  \begin{eqnarray*}
    {csd \choose ks} / {nd \choose ks } \leq \left(\frac{(csd)^{ks}}{(ks)!}\right) / \left(\frac{(nd-ks)^{ks}}{(ks)!}\right) \leq \left(\frac{csd}{nd-ks}\right)^{ks}\leq \left(\frac{csd}{(d-k\gamma)n}\right)^{ks}.
  \end{eqnarray*}
  
  For a fixed s, $X$ can be chosen in ${n \choose cs}$ ways and $S$ can be chosen in ${dn \choose s}$ ways.
  Thus, the probability that such an event occurs is upper-bounded by 
  \begin{eqnarray*}
    &&{n \choose cs}{dn \choose s}\left(\frac{csd}{(d-k\gamma)n}\right)^{ks}\\
    &\leq& \left(\frac{en}{cs}\right)^{cs}\left(\frac{edn}{s}\right)^s\left(\frac{csd}{(d-k\gamma)n}\right)^{ks}\\
    &\leq& \left[\left(\frac{s}{n}\right)^{\eta} e^{k-\eta}c^{1+\eta}d^{k+1}(d-k\gamma)^{-k}\right]^s \\
    &\leq& \left(\frac{s \beta}{n}\right)^{\eta s}
  \end{eqnarray*}
  for some $\beta$.

  By summing over $1\leq s\leq \gamma n$,
  \begin{eqnarray*}
    \sum_{s=1}^{\gamma n}\left(\frac{s \beta}{n}\right)^{\eta s} 
    \leq \sum_{s=1}^{\log n}\left(\frac{s \beta}{n}\right)^{\eta s} +\sum_{s=\log n+1}^{\gamma n}\left(\frac{s\beta}{n}\right)^{\eta s} 
    = O\left(\frac{\beta^\eta \log n}{n^\eta}\right)+O\left((\gamma \beta^\eta)^{\eta \log n}\right).
  \end{eqnarray*}
  The first term is $o(1)$ and the second term is also $o(1)$ by taking $\gamma$ small enough.
\end{proof}

\subsection{Proof of Lemma~\ref{lmm:far}}\label{apx:far}
\begin{proof}
  Let us fix an assignment $x\in \bit^n$ over variables and $\bfX_e$ be a random variable indicating that the constraint $C_e$ is satisfied by the assignment.
  Then, $E[\bfX_e]=|P^{-1}(1)|/2^k$ and all $\bfX_e$ are mutually independent since every $\bfb_e$ is mutually independent.
  Let $\bfX=\sum\bfX_e$, 
  then from Hoeffding's inequality, $\Pr[|\bfX-E[\bfX]|\leq \epsilon E[\bfX] ] < \exp(-\Omega(\epsilon^2 d n))$.
  By choosing $d=\Omega(1/\epsilon^2)$, 
  the union bound over all $2^n$ possible assignments over variables yields the desired results.
\end{proof}

\subsection{Proof of Lemma~\ref{lmm:few-cycles}}\label{apx:few-cycles}
\begin{proof}
  After the $t$-th interaction,
  the number of vertices in the query-answer history is at most $kt$.
  Thus, the sum of the remaining degrees of those vertices is at most $dkt$.
  On the other hand, the sum of the remaining degrees of other vertices is at least $dn-dkt$.
  Thus, the probability that the $i$-th vertex $(1\leq i \leq k)$ at the edge for the $t$-th answer is contained in the query-answer history is at most $dkt/(dn-dkt) \leq 2dkt/dn$ when $t\leq n/2k$.
  Therefore, the expected number of $\cy(G)$ is at most 
  \begin{eqnarray*}
    \sum_{t=1}^{O(n^{1/2+\delta})}k\cdot \frac{2dkt}{dn} \leq O(k^2n^{2\delta}).
  \end{eqnarray*}
  From Markov's inequality, the lemma follows.
\end{proof}
\subsection{Proof of Lemma~\ref{lmm:large-girth}}\label{apx:large-girth}
\begin{proof}
  Let $q_t=(v_t,i_t)$ be the $t$-th query by $\cala$.
  After the $t$-th interaction,
  the number of vertices in the query-answer history is at most $kt$.
  Since the degree is bounded by $d$,
  the number of vertices in the query-answer history whose distance from $v_t$ is at most $g$ is at most $(dk)^g$.
  Thus, the sum of the remaining degrees of such vertices is at most $d(dk)^g$.
  On the other hand, the sum of the remaining degrees of other vertices is at least $dn-d(dk)^g$.
  Thus, the probability that the $i$-th vertex $(1\leq i \leq k)$ of the edge for the $t$-th answer is contained in the query-answer history is at most $d(dk)^g/(dn-d(dk)^g) \leq 2d(dk)^g/dn$.
  The last inequality is from $g\leq (\log_{dk}n)/2$.
  Therefore, by union bound, the probability that such an event occurs is at most,
  \begin{eqnarray*}
    k \frac{2d(dk)^g}{dn}O(n^{1/2+\delta})=O(2kn^{-\delta})=o(1).
  \end{eqnarray*}
\end{proof}

\subsection{Proof of Lemma~\ref{lmm:cannot-guess-tree}}\label{apx:cannot-guess-tree}
First, for an edge $e$ and a vertex $v\in e$,
we show that $\bfx_v$ and $\bfx_e$ are independent.
\begin{lemma}\label{lmm:cannot-guess-edge}
  Let $H$ be a hypergraph and let $e$ be an edge of $H$.
  Then, for any vertex $v\in e$, $\bfx_v$ and $\bfb_e$ are independent under $\cald_H$.
\end{lemma}
\begin{proof}
  We show that $\bfb_e$ is uniform after we choose the value of $\bfx_v$.
  Let $e=(v_1,\ldots,v_k)$ and we assume that $\bfx_{v_1}=0$ without loss of generality.
  Then, $\cald_H$ generates $\bfx_{v_2},\ldots,\bfx_{v_k}$ uniformly at random.
  Let $x \in \bit^{k-1}$ be the vector of chosen values.
  Then, $\cald_H$ chooses $\bfb_e$ from the set $S_x=\{b\mid P((0,x)+b)=1\}$.
  Let $s=|P^{-1}(1)|$ be the size of $S$.
  Here, we separate $P^{-1}(1)$ into $s/2$ couples of vectors $(p,\overline{p})$.
  Let $(p_1,\overline{p}_1),\ldots,(p_{s/2},\overline{p}_{s/2})$ be the set of such couples.
  Then, $S_x$ is partitioned into $S_{x,i}=\{b\mid(0,x)+b=p_i\mbox{ or }(0,x)+b=\overline{p}_i\} (1\leq i \leq s/2)$.
  We consider the set $S_i=\bigcup_{x\in \bit^{k-1}}S_{x,i}$.
  Then, it is easy to see that $S_i=\bit^k$, i.e., every vector from $\bit^k$ appears exactly once in $S_{x,i} (1\leq i\leq s/2)$.
  Thus, eventually, $\bfb_e$ is distributed uniformly at random in $\bit^k$.
\end{proof}

\begin{proof}[Proof of Lemma~\ref{lmm:cannot-guess-tree}]
  We use induction on the number of edges of $T$.
  When $T$ consists of one edge, 
  the lemma holds from Lemma~\ref{lmm:cannot-guess-edge}.

  Let $m\geq 2$ be an integer.
  Assume that, 
  for any hypertree $T=(V,E)$ with $|E|<m$ and $v\in V$,
  $\bfx_{v}$ and $\bfb_{E}$ are independent under $\cald_H$.
  Let $T=(V,E)$ be a hypertree with $|E|=m$ and $v$ be the supposed vertex.
  Since $m\geq 2$,
  there exists an edge $e$ such that $e$ contains a leaf,
  but does not contain $v$ as a leaf.
  Let $w$ be the unique vertex that connects $T-e$ and $e$.
  Note that $w$ may coincides with $v$.
  Then, 
  \begin{eqnarray*}
    && 
    \Pr[\bfb_e|\bfb_{E-e}]\\
    &=&
    \sum_{\bfx_w}\Pr[\bfb_e|\bfx_w]\Pr[\bfx_w|\bfb_{E-e}]\;\;\;(\mbox{from }\bfb_e \ci \bfb_{E-e} \mid \bfx_w) \\
    &=&
    \sum_{\bfx_w}\Pr[\bfb_e]\Pr[\bfx_w|\bfb_{E-e}]\;\;\;(\mbox{from Lemma}~\ref{lmm:cannot-guess-edge}) \\
    &=&
    \Pr[\bfb_e].
  \end{eqnarray*}
  Thus, 
  \begin{eqnarray*}
    &&
    \Pr[\bfx_v|\bfb_E]
    =
    \sum_{\bfx_w}\frac{\Pr[\bfx_v,\bfx_w,\bfb_E]}{\Pr[\bfb_E]}
    =
    \sum_{\bfx_w}\frac{\Pr[\bfx_v,\bfb_{E-e}|\bfx_w,\bfb_e]\Pr[\bfb_e|\bfx_w]\Pr[\bfx_w]}{\Pr[\bfb_{E-e}]\Pr[\bfb_{e}|\bfb_{E-e}]}.
  \end{eqnarray*}
  Since \(\Pr[\bfx_v,\bfb_{E-e}|\bfx_w,\bfb_e]=\Pr[\bfx_v,\bfb_{E-e}|\bfx_w]\), \(\Pr[\bfb_e|\bfx_w]=\Pr[\bfb_e]\) from Lemma~\ref{lmm:cannot-guess-edge} and \(\Pr[\bfb_e|\bfb_{E-e}]=\Pr[\bfb_e]\), this is equal to
  \begin{eqnarray*}
    \sum_{\bfx_w}\frac{\Pr[\bfx_v,\bfb_{E-e}|\bfx_w]\Pr[\bfb_e]\Pr[\bfx_w]}{\Pr[\bfb_{E-e}]\Pr[\bfb_{e}]}
    =
    \sum_{\bfx_w}\frac{\Pr[\bfx_v,\bfx_w,\bfb_{E-e}]}{\Pr[\bfb_{E-e}]}
    =
    \Pr[\bfx_v|\bfb_{E-e}].
  \end{eqnarray*}
  From the assumption of the induction, we have \(\Pr[\bfx_v|\bfb_E]=\Pr[\bfx_v]\).
\end{proof}

\subsection{Proof of Lemma~\ref{lmm:merge-vertex}}\label{apx:merge-vertex}
\begin{lemma}\label{lmm:product-tv}
  Let $\cald$ and $\cald_i (1\leq i\leq k)$ be distributions generating $\bfx\in \bit$.
  And, suppose that $\Pr_{\cald}[\bfx=x]$ is given by
  \begin{eqnarray*}
    \Pr_{\cald}[\bfx=x]=\frac{\prod_{i=1}^{k}\Pr_{\cald_i}[\bfx=x]}{\prod_{i=1}^{k}\Pr_{\cald_i}[\bfx=0]+\prod_{i=1}^{k}\Pr_{\cald_i}[\bfx=1]}.
  \end{eqnarray*}
  Then, $\dtv[\cald(\bfx)]\leq \sum_{i=1}^{k}\dtv[\cald_i(\bfx)]$.
  The equality only holds when (at least) $k-1$ of $\dtv[\cald_i(\bfx)](1\leq i \leq k)$ are $0$.
\end{lemma}
\begin{proof}
  We use induction on $k$.
  When $k=1$, we have nothing to prove.

  Suppose that the lemma holds when $k<t$.
  Now, we show that the lemma also holds when $k=t$.
  For notational simplicity,
  we define $p_x^i=\prod_{j=1}^{i}\Pr_{\cald_j}[\bfx=x]$.
  Then, $\Pr_{\cald}[\bfx=x]=p_x^t/(p_x^t+p_{1-x}^t)$.
  Let $\cald'$ be a distribution generating $\bfx\in \bit$ in such a way that
  \begin{eqnarray*}
    \Pr_{\cald'}[\bfx=x]=\frac{p_x^{t-1}}{p_x^{t-1}+p_{1-x}^{t-1}}.
  \end{eqnarray*}
  From the assumption, we have $\dtv[\cald'(\bfx)]\leq \sum_{i=1}^{t-1}\dtv[\cald_i(\bfx)]$.
  Also,
  \begin{eqnarray*}
    \Pr_{\cald}[\bfx=x]
    &=&
    \frac{p_x^t}{p_x^t+p_{1-x}^t}\\
    &=&
    \frac{p_x^{t-1}\Pr_{\cald_t}[\bfx=x]}{p_x^{t-1}\Pr_{\cald_t}[\bfx=x]+p_{1-x}^{t-1}\Pr_{\cald_t}[\bfx=1-x]}\\
    &=&
    \frac{\Pr_{\cald'}[\bfx=x]\Pr_{\cald_t}[\bfx=x]}{\Pr_{D'}[\bfx=x]\Pr_{\cald_t}[\bfx=x]+\Pr_{D'}[\bfx=1-x]\Pr_{\cald_t}[\bfx=1-x]}.\\
  \end{eqnarray*}
  Let $\delta'=\Pr_{\cald'}[\bfx=0]-1/2$ and $\delta_t=\Pr_{\cald_t}[\bfx=0]-1/2$ where $-1/2\leq \delta',\delta_t\leq 1/2$.
  Then, 
  \begin{eqnarray}
    \left| \Pr_{\cald}[\bfx=0]-\frac{1}{2} \right| 
    &=&
    \left| \frac{(\frac{1}{2}+\delta')(\frac{1}{2}+\delta_t)}{(\frac{1}{2}+\delta')(\frac{1}{2}+\delta_t)+(\frac{1}{2}-\delta')(\frac{1}{2}-\delta_t)}-\frac{1}{2} \right| \nonumber\\
    &=&
    \left| \frac{\delta'+\delta_t}{1+4\delta'\delta_t} \right| 
    \leq
    |\delta'|+|\delta_t| \label{eqn:leq}.
  \end{eqnarray}
  Similarly, $\left|\Pr_{\cald}[\bfx=1]-1/2 \right| \leq |\delta'|+|\delta_t|$.
  Thus, 
  $\dtv[\cald(\bfx)] \leq 2(|\delta'|+|\delta_t|)\leq \dtv[\cald'(\bfx)]+\dtv[\cald_t(\bfx)] \leq \sum_{i=1}^t\dtv[\cald_i(\bfx)]$.
  
  We finally remark about the condition that the equality holds.
  We note that the case $\delta'=1/2$ and $\delta_t=-1/2$ or vice versa cannot happen since in this case we cannot decide the value of $\bfx$.
  Thus, the equality of (\ref{eqn:leq}) only holds when $\delta'$ or $\delta_t$ equals zero.
  Also, $\dtv[\cald(\bfx)]$ becomes non-zero if $\dtv[\cald'(\bfx)]$ or $\dtv[\cald_t(\bfx)]$ is non-zero.
  Thus, the claim holds.
\end{proof}

\begin{proof}[Proof of Lemma~\ref{lmm:merge-vertex}]
  We note that $\bfx_{L_i} \ci \bfx_{L_j} \mid \bfx_v$ for $1\leq i,j\leq \ell$ under $\cald_H(\cdot|\bfb_E)$.
  Since $\cald_H(\bfx_v|\bfb_E)$ is uniform from Lemma~\ref{lmm:cannot-guess-tree},
  we have from Lemma~\ref{lmm:product} that
  \begin{eqnarray*}
    \Pr[\bfx_v=x|\{\bfx_{L_i} \}_{i=1}^\ell,\bfb_E]=\frac{\prod_{i=1}^\ell\Pr[\bfx_v=x|\bfx_{L_i},\bfb_E]}{\sum_{x'}\prod_{i=1}^\ell\Pr[\bfx_v=x'|\bfx_{L_i},\bfb_E]}.
  \end{eqnarray*}
  Then, by applying Lemma~\ref{lmm:product-tv}, we have the desired result.
\end{proof}

\subsection{Proof of Lemma~\ref{lmm:merge-edge}}\label{apx:merge-edge}
\begin{proof}
  Let $p_x=\Pr[\bfx_e=x|\{\bfx_{L_i}\}_{i=1}^{k-1},\bfb_E]$ for $x\in \bit^{k}$.
  Also, we let \(\delta_i=\Pr[\bfx_{v_i}=0|\bfx_{L_i},\bfb_E]-1/2\) where $-1/2\leq \delta_i \leq 1/2$ and \(S=\supp(\cald_H(\bfx_e|\bfb_E))\subseteq \bit^k \).
  Then, from Lemma~\ref{lmm:product}, 
  \begin{eqnarray*}
    p_x=\frac{\prod_{i=1}^{k-1}(1/2+(-1)^{x_i}\delta_i)}{\sum_{x'\in S}\prod_{i=1}^{k-1}(1/2+(-1)^{x'_i}\delta_i)}.
  \end{eqnarray*}
  Let \(S_v=\{x \in S | x_k=v \}\), then 
  \begin{eqnarray*}
    \dtv[\cald_H(\bfx_{v_k}|\{\bfx_{L_i}\}_{i=1}^{k-1},\bfb_E)]
    =
    \left| \sum_{x\in S_1}p_x-\sum_{x\in S_0}p_x \right| 
    \leq 
    \sum_{x\in S_1} |p_x-p_{\overline{x}}|
    =
    \sum_{x\in S_1} (p_x+p_{\overline{x}})\frac{|p_x-p_{\overline{x}}|}{p_x+p_{\overline{x}}}.
  \end{eqnarray*}
  For each $x\in S_1$, 
  we can think of a random (Boolean) variable that takes $1$ with probability $p_x/(p_x+p_{\overline{x}})$ and $0$ with probability $p_{\overline{x}}/(p_x+p_{\overline{x}})$.
  Then, $|p_x-p_{\overline{x}}|/(p_x+p_{\overline{x}})$ can be seen as the total variation distance of this random variable.
  Since the probability distribution of this random variable exactly matches the condition of Lemma~\ref{lmm:product-tv} (note that the denominator of $p_x$ is canceled out), 
  we have
  \begin{eqnarray}
    \dtv[\cald_H(\bfx_{v_k}|\{\bfx_{L_i}\}_{i=1}^{k-1},\bfb_E)]
    &\leq &
    \sum_{x\in S_1} (p_x+p_{\overline{x}}) \sum_{i=1}^{k-1}\dtv[\cald_H(\bfx_{v_i}|\bfx_{L_i},\bfb_E)] \label{eqn:leq2} \\
    &=&
    \sum_{i=1}^{k-1}\dtv[\cald_H(\bfx_{v_i}|\bfx_{L_i},\bfb_E)]. \nonumber
  \end{eqnarray}
  This already indicates that $\rho(P)\leq 1$.

  We fix $P\neq \equ$ and let $\delta=(\delta_1,\ldots,\delta_{k-1})$.
  Let $D_\epsilon=\{\delta\in[-1/2,1/2]^{k-1}\mid \sum_{i}^{k-1}|\delta_i| \leq 1+\epsilon  \}$ for $\epsilon>0$ chosen later.
  This excludes singular points implied by the left hand side of (\ref{eqn:leq}).
  We define
  \begin{eqnarray*}
    \rho(\delta)
    &=&
    \frac{\dtv[\cald_H(\bfx_{v_k}|\{\bfx_{L_i}\}_{i=1}^{k-1},\bfb_E)]}{\sum_{i=1}^{k-1}\dtv[\cald_H(\bfx_{v_i}|\bfx_{L_i},\bfb_E)]}\\
    &=&
    \frac{\left| \sum_{x\in S_1}p_x-\sum_{x\in S_0}p_x \right|}{\sum_{i=1}^{k-1}2|\delta_i|}.
  \end{eqnarray*}
  We can safely state that $\rho(\delta)<1/(1+\epsilon)$ in $[-1/2,1/2]^{k-1}-D_\epsilon$.

  Next, we only consider the domain $D_\epsilon^{+}=D_\epsilon\cap[0,1/2]^{k-1}$.
  Other domains (i.e., $D_\epsilon-[0,1/2]^{k-1}$) can be treated in the same manner.
  After a calculation,
  we can see that the limit at $\delta=(0,\ldots,0)$ exists and the value is less than one when $P\neq \equ$.
  Thus, $\rho(\delta)$ is continuous in $D_\epsilon^{+}$.
  In particular, $\rho(\delta)$ is uniformly continuous.

  We will show that there exists a universal constant $\rho<1$ such that $\rho(\delta)\leq \rho$ regardless of $\delta \in D_{\epsilon}^{+}$.
  This concludes the lemma.
  For $\epsilon>0$, 
  we define $H_\epsilon=D_{\epsilon}^{+}\cap \{\delta\in [0,1/2]^{k-1}\mid \mbox{all but one of }\delta_i \leq \epsilon \}$.
  When $\delta \not \in  H_{\epsilon}$, 
  considering the inequality (\ref{eqn:leq2}) in the proof of Lemma~\ref{lmm:product-tv},
  there exists some constant $\rho<1$ such that $\rho(\delta)\leq \rho$.

  The remaining case is $\delta\in H_\epsilon$.
  We will show that there exists a constant $\rho<1$ such that $\rho(\delta)\leq \rho$ for $\delta\in H_0$.
  From the uniform continuity of $\rho(\delta)$, 
  by choosing $\epsilon>0$ small enough,
  we establish the desired result.

  Let $\delta\in H_0$.
  Without loss of generality,
  we can assume that $\delta_1\geq 0$ and $\delta_2=\cdots=\delta_{k-1}=0$.
  For $\rho \geq 0$, 
  we consider the following function of $\rho$ and $\delta$.
  \begin{eqnarray*}
    f(\rho,\delta)
    &=&
    \rho \sum_{i=1}^{k-1}\dtv[\cald_H(\bfx_{v_i}|\bfx_{L_i},\bfb_E)]  - \dtv[\cald_H(\bfx_{v_k}|\{\bfx_{L_i}\}_{i=1}^{k-1},\bfb_E)] \\
    &=&
    \rho \sum_{i=1}^{k-1}2|\delta_i|-\left| \sum_{x\in S_1}p_x-\sum_{x\in S_0}p_x \right|\\
    &=&
    \min\left(\rho \sum_{i=1}^{k-1}2\delta_i-\left(\sum_{x\in S_1}p_x-\sum_{x\in S_0}p_x\right),\rho\sum_{i=1}^{k-1}2\delta_i+\left(\sum_{x\in S_1}p_x-\sum_{x\in S_0}p_x\right)\right)\\
    &=:& 
    \min \left(f_1(\rho,\delta),f_2(\rho,\delta)\right).
  \end{eqnarray*}
  We let $g(\rho,\delta_1)=f(\rho,\delta_1,0,\ldots,0)$.
  If the minimum of $g(\rho,\delta_1)$ over $0\leq \delta_1 \leq 1/2$ is non-negative, 
  we can say that $\rho(\delta)\leq \rho$ for $\delta\in H_0$.
  Since the denominator of $g$ (after factoring) is non-negative,
  we only consider its numerator $\hat{g}=:\min\{\hat{g}_1,\hat{g}_2\}$.
  We note that $\hat{g}_1,\hat{g}_2$ are odd functions and the degrees of $\hat{g}_1,\hat{g}_2$ are at most two.
  Thus, $\hat{g}_1,\hat{g}_2$ is a linear function of $\delta_1$ when we fix $\rho$.
  Suppose that $\hat{g}(\rho,\delta_1)<0$ for some $\delta_1\in [0,1/2]$.
  Then, by moving $\delta_1$ to $0$ or $1/2$,
  we obtain a smaller value.
  Thus, it suffices to check the case $\delta_1=0$ and $\delta_1=1/2$.
  When $\delta_1=0$, we have already seen that $\rho(0,\ldots,0)<1$.
  The case $\delta_1=1/2$ corresponds to the following question:
  how much can we guess the value of $\bfx_{v_k}$ when we know the actual value of $\bfx_{v_1}$ and we do not know values of other variables?
  For any symmetric predicate except \equ, 
  the choice of $\bfx_{v_k}$ is not unique.
  Thus, we can choose $\rho<1$ that only depends on $P$.
\end{proof}

\section{An $\Omega(n^{1/2+\delta})$ Lower Bound for Testing \txor}\label{apx:txor}
In this section, 
we give the proof of Theorem~\ref{thr:txor}.
To make hard instances that are $\epsilon$-close to satisfiability,
we slightly modify the construction of $\caldsat$.
We use the same $\caldfar$ as defined in Section~\ref{sec:two-sided}.
\begin{definition}
  Let $H=(V,E)$ be a graph with $n$ vertices.
  Let $\epsilon$ be an error parameter.
  Define a distribution $\cald_{H,\epsilon}$ generating an instance $\Phi$ of \txor as follows.
  The variable set of $\Phi$ is $\{x_v\}_{v\in V}$.
  We choose $\bfx\in \bit^n$ uniformly at random.
  For each edge $e=(u,v)\in E$, 
  we choose $\bfb_e$ uniformly at random from the set $\{b\in \bit^2\mid P((\bfx_{u},\bfx_{v})+b)=1 \}$ with probability $1-\epsilon$,
  and from the set $\{b\in \bit^2\mid P((\bfx_{u},\bfx_{v})+b)=0 \}$ with probability $\epsilon$.
  Then, we add a constraint $C_e$ of the form $P((x_{u},x_{v})+\bfb_e)=1$ to $\Phi$.
\end{definition}
\begin{definition}
  Given parameters $n,d,\epsilon$, 
  define a distribution ${\caldsat}$ generating an instance of \txor as follows.
  First, we choose a graph $H$ from $\calg_{n,d,2}$.
  Then, an instance is output according to $\cald_{H,\epsilon}$.
\end{definition}
The following lemma is immediate.
\begin{lemma}\label{lmm:close}
  For any $\epsilon>0$ and $d\geq 1$, the following holds.
  Let $\Phi$ be an instance of \txor chosen from $\caldsat$.
  Then, $\Phi$ is $(\epsilon/2)$-close to satisfiability with a probability of $1-o(1)$.
  \qed
\end{lemma}

We use $\caldsat$ defined above instead of $\caldsat$ defined in Section~\ref{sec:two-sided} to prove Theorem~\ref{thr:txor}.
The proof is almost same as the proof of Theorem~\ref{thr:two-sided}.
A modification occurs only in the proof of Lemma~\ref{lmm:merge-edge}.
The following is an analogue of Lemma~\ref{lmm:merge-edge} for the distribution $\cald_{H,\epsilon}$
\begin{lemma}\label{lmm:merge-edge-txor}
  Let $T=(V,E)$ be a subgraph of a graph $H$.
  Suppose that $T$ is a tree.
  Let $T_u$ be the subtree obtained by removing $e=(u,v)\in E$.
  Here, $u$ is the root of $T_u$.
  Let $L$ be a subset of the leaves of $T_u$.
  Then,
  \begin{eqnarray*}
    \dtv[\cald_{H,\epsilon}(\bfx_{v}|\bfx_L,\bfb_E)]\leq (1-2\epsilon)\dtv[\cald_{H,\epsilon}(\bfx_u|\bfx_{L},\bfb_E)].
  \end{eqnarray*}
\end{lemma}
\begin{proof}
  For simplicity, 
  we assume that $\bfb_e=0$.
  We can prove other cases in the same manner.
  It holds that 
  \begin{eqnarray*}
    \Pr[\bfx_v=0|\bfx_L,\bfb_E]
    &=&
    \sum_{\bfx_u}\Pr[\bfx_v=0|\bfx_u,\bfb_E]\Pr[\bfx_u|\bfx_L,\bfb_E]\\
    &=&
    (1-\epsilon)\Pr[\bfx_u=1|\bfx_L,\bfb_E]+\epsilon\Pr[\bfx_u=0|\bfx_L,\bfb_E].
  \end{eqnarray*}
  Let $\delta=\Pr[\bfx_u=0|\bfx_L,\bfb_E]-1/2$.
  Then,
  \begin{eqnarray*}
    \left| \Pr[\bfx_v=0|\bfx_u,\bfb_E]-\frac{1}{2} \right| 
    &=&
    \left| (1-\epsilon)(\frac{1}{2}-\delta)+\epsilon(\frac{1}{2}+\delta)-\frac{1}{2} \right|\\
    &=& 
    (1-2\epsilon)|\delta|.
  \end{eqnarray*}
  Similarly, $\left|\Pr[\bfx_v=1|\bfx_L,\bfb_E]-1/2\right| \leq (1-2\epsilon)|\delta|$.
  Thus, 
  $\dtv[\cald_{H,\epsilon}(\bfx_v|\bfx_L,\bfb_E)] \leq 2(1-2\epsilon)|\delta| \leq (1-2\epsilon)\dtv[\cald_{H,\epsilon}(\bfx_u|\bfx_L)]$.
\end{proof}
\begin{proof}[Proof of Theorem~\ref{thr:txor}]
  Combining the proof of Theorem~\ref{thr:two-sided} and Lemma~\ref{lmm:merge-edge-txor}, 
  the theorem holds.
  Since we use $1-2\epsilon$ instead of $\rho(P)$ as a decaying factor,
  we need to choose $\delta=O(\epsilon / \log (k/\epsilon^2))$.
\end{proof}

\section{A Linear Lower Bound for One-Sided Error Testers}\label{apx:one-sided}
\begin{theorem}\label{thr:one-sided-weak}
  Let $P:\bit^k\to\bit$ be any symmetric predicate except \equ where $k\geq 3$.
  Then,
  for any $\epsilon>0$, 
  there exists $d\geq 1$ such that any one-sided error $(|P^{-1}(0)|/2^k-\epsilon)$-tester for \csp{$P$} with a degree bound $d$ requires $\Omega(n)$ queries.
\end{theorem}
\begin{proof}
  Let $\Phi$ be a given instance.
  Since a one-sided error tester must accept $\Phi$ when $\Phi$ is satisfiable,
  it cannot reject $\Phi$ unless it has found an unsatisfiable sub-instance of $\Phi$.
  We show that for any $\epsilon>0$ there exists $d$ for which the following holds: 
  there exists an instance $\Phi$ of \csp{$P$} with a degree bound $d$ such that any linear-size sub-instance is satisfiable while $\Phi$ is $(|P^{-1}(0)|/2^k-\epsilon)$-far from satisfiability.
  The lemma clearly holds from this fact.

  From Lemma~\ref{lmm:far}, 
  for any $\epsilon>0$ and $\eta>0$,
  there exist $d\geq 1, \gamma>0$, and an instance $\Phi$ with a degree bound $d$ such that $\Phi$ is $(|P^{-1}(0)|/2^k-\epsilon)$-far from satisfiability and the underlying hypergraph is a $(\gamma,\eta)$-expander.
  Let $\Phi'$ be a sub-instance of $\Phi$,
  and let $V(\Phi')$ and $E(\Phi')$ denote the set of variables and constraints of $\Phi'$, respectively.
  We show that $\Phi'$ is satisfiable when $|E(\Phi')|\leq \gamma n$ by induction on $|E(\Phi')|$.

  Clearly, any sub-instance with no constraint is satisfiable.
  Suppose that any sub-instance of $\Phi$ with less than $m$ constraints is satisfiable.
  Let $\Phi'$ be a sub-instance of $\Phi$ with $m$ constraints.
  Then, since $H$ is a $(\gamma,\eta)$-expander,
  $|V(\Phi')|\geq (k-1-\eta)|E(\Phi')|$.
  Since $\eta<1$, 
  there exists some constraint $C\in E(\Phi')$ such that $C$ shares at most two variables with $E(\Phi')-C$.
  Suppose that $C$ shares two variables $x_u,x_v$ with $E(\Phi')-C$.
  Note that $P$ is a symmetric predicate except \equ.
  If $P$ accepts $x\in \bit^k$ with $|x|=1$, 
  then $P$ accepts $x$ with $|x|=k-1$.
  If not, there exists some $2\leq w \leq k-2$ such that $P$ accepts $x$ with $|x|=w$.
  Thus, $P$ accepts $x\in \bit^k$ with $|x|=w$ for some $2\leq w \leq k-1$.
  Hence, regardless of the values of $x_u,x_v$,
  we can satisfy $C$ by appropriately choosing the values of the rest of the variables in $C$.
  Other cases are similar.
  Thus, the induction completes and the theorem follows.
\end{proof}
\begin{proof}[Proof of Theorem~\ref{thr:one-sided}]
  Let $Q:\bit^k\to\bit$ be a predicate such that $P^{-1}(1)\subseteq Q^{-1}(1)$ for a symmetric predicate $P\bit^k\to\bit$ except \kequ.
  As the proof of Theorem~\ref{thr:two-sided},
  We change the definition of $\cald_H$.  
  Thus, 
  for each edge $e=(v_1,\ldots,v_k)$ of $H$,
  we choose $\bfb_e$ uniformly at random from the set $\{b\in \bit^k\mid P((\bfx_{v_1},\ldots,\bfx_{v_k})+b)=1 \}$ instead of $\{b\in \bit^k\mid Q((\bfx_{v_1},\ldots,\bfx_{v_k})+b)=1 \}$.
  Then, the rest of the proof is the same as the proof of Theorem~\ref{thr:one-sided-weak}.
\end{proof}

\section{An $\epsilon$-Tester for \kequ}\label{apx:equ}
In this section, we prove Theorem~\ref{thr:equ}.
The idea is to transform an instance of \kequ into a graph and use a bipartiteness tester given in~\cite{GR99}.

We define a reduction $\varphi$, which maps an instance $\Phi$ of \kequ to an instance $\Phi'$ of $2$-\equ.
The set of variables of $\Phi'$ is the same as $\Phi$.
For each constraint in $\Phi$ of the form $\ell_1=\ell_2=\ldots=\ell_k$,
where each $\ell_i$ is a literal,
we simply introduce $k(k-1)/2$ constraints in $\Phi'$ of the form $\ell_i=\ell_j (1\leq i \leq j \leq k)$.

\begin{lemma}\label{lmm:reduction-from-keq}
  If $\Phi$ is a satisfiable instance of \kequ,
  then $\varphi(\Phi)$ is satisfiable.
  On the contrary, if $\Phi$ is $\epsilon$-far from satisfiability,
  then $\varphi(\Phi)$ is $\epsilon'$-far from satisfiability where $\epsilon'=2\epsilon/k$.
\end{lemma}
\begin{proof}
  Let $\Phi'=\varphi(\Phi)$.
  The former part is obvious.
  Furthermore, if $\Phi'$ is satisfiable, then $\Phi$ is satisfiable.

  We show the latter part.
  Suppose that $\Phi'$ is not $\epsilon'$-far from satisfiability.
  Since the degree bound of $\Phi'$ is (at most) $d'=dk$,
  we can make $\Phi'$ satisfiable by removing less than $\epsilon' d'n/2=\epsilon dn / k$ constraints.
  Let $\Phi'_{\textrm{rm}}$ be the resulting instance.

  We simulate this removal in $\Phi$.
  That is, for each removed constraint in $\Phi'$,
  we remove the corresponding constraint in $\Phi$.
  Let $\Phi_{\textrm{rm}}$ be the resulting instance of \kequ.
  The number of removed constraints is at most $\epsilon dn/k$.
  The important fact is that $\varphi(\Phi_{\textrm{rm}})$ is a sub-instance of $\Phi'_{\textrm{rm}}$.
  Since $\Phi'_{\textrm{rm}}$ is satisfiable, 
  $\Phi_{\textrm{rm}}$ is also satisfiable.
  However, this contradicts the fact that $\Phi$ is $\epsilon$-far from satisfiability.
\end{proof}

Next,
we define a reduction $\varphi_G$, which maps an instance $\Phi$ of $2$-\equ to a graph $G$.
First, each literal of $\Phi$ forms a vertex in $G$.
Next, for each variable $x$ of $\Phi$, 
we introduce an edge $(x,\overline{x})$ in $G$.
We call these edges \textit{variable edges}.
Furthermore, 
for each constraint in $\Phi$ of the form $\ell_1=\ell_2$ where $\ell_1$ and $\ell_2$ are literals,
we introduce two edges $(\ell_1,\overline{\ell_2})$ and $(\overline{\ell_1},\ell_2)$ in $G$.
We call these edges \textit{constraint edges}.
The supposed bipartition of $G$ is into the set of literals whose values are $1$ (true) and $0$ (false) in the solution of $\Phi$.

\begin{lemma}\label{lmm:reduction-from-2eq}
  If $\Phi$ is a satisfiable instance of \textsf{$2$-\equ},
  then $\varphi_G(\Phi)$ is satisfiable.
  On the contrary, if $\Phi$ is $\epsilon$-far from satisfiability,
  then $\varphi_G(\Phi)$ is $\epsilon'$-far from satisfiability where $\epsilon'=\epsilon/(4d)$.
\end{lemma}
\begin{proof}
  Let $G=(V,E)=\varphi_G(\Phi)$.
  The number of vertices of $G$ is $2n$,
  where $n$ is the number of variables of $\Phi$ and the degree bound $d$ of $G$ is the same as $\Phi$.
  The former part of the lemma is obvious.
  Furthermore, if $G$ is bipartite, then $\Phi$ is satisfiable.
  
  We show the latter part.
  Suppose that $G$ is not $\epsilon'$-far from satisfiability.
  Let $E'\subseteq E$ be the set of edges such that $G$ becomes bipartite by removing $E'$ and $|E'|<\epsilon' d (2n)/2$.
  First, we canonicalize $E'$ so that $E'$ does not contain variable edges.
  This is done as follows.
  If $E'$ contains a variable edge $(x,\overline{x})$,
  we exclude the edge from $E'$,
  and instead we add to $E'$ every constraint edge of the form $(x,\ell)$ and $(\overline{x},\ell)$ where $\ell$ is a literal.
  This preserves the property that $G$ becomes bipartite by removing $E'$.
  Since the degree bound of $G$ is $d$,
  after canonicalizing $E'$, 
  the size of $|E'|$ is at most $2d\cdot \epsilon' d (2n)/2=\epsilon dn/2$.
  Let $G_{\textrm{rm}}$ be the resulting graph after removing $E'$.

  We simulate this removal in $\Phi$.
  That is, for each removed edge in $G$,
  we remove the corresponding constraint in $\Phi$.
  This can be done since we excluded variable edges.
  Let $\Phi_{\textrm{rm}}$ be the resulting instance of $2$-\equ.
  The number of removed constraints is at most $\epsilon dn/2$.
  Again, $\varphi(\Phi_{\textrm{rm}})$ is a sub-instance of $G_{\textrm{rm}}$.
  Since $G_{\textrm{rm}}$ is bipartite, 
  $\Phi_{\textrm{rm}}$ is satisfiable.
  However, this contradicts the fact that $\Phi$ is $\epsilon$-far from satisfiability.
\end{proof}

Finally, we use the following algorithm for testing bipartiteness.
\begin{lemma}\cite{GR99}\label{lmm:bipartiteness}
  There exists a one-sided error $\epsilon$-tester for bipartiteness whose running time is $O(\sqrt{n}\poly(\log n/\epsilon))$,
  where $n$ is the number of vertices.
\end{lemma}
\begin{proof}[Proof of~Theorem~\ref{thr:equ}]
Combining Lemmas~\ref{lmm:reduction-from-keq},~\ref{lmm:reduction-from-2eq}, and~\ref{lmm:bipartiteness}, 
the theorem holds.
\end{proof}

\section{A Linear Lower Bound for Testing \kcsp}\label{apx:kcsp}
In this section, 
we show that there exists a certain predicate $P$ such that,
\csp{$P$} requires linear number of queries to distinguish satisfiable instances from instances $(1-2k/2^k-\epsilon)$-far from satisfiability.
Then, we show the hardness of \mis.
We use a matrix to define the predicate.
\begin{definition}
  For a matrix $A\in \bit^{h\times k}$, a predicate $P_A:\bit^k\to\bit$ is defined as
  \begin{eqnarray*}
    P_A(x_1,\ldots,x_k)=1  \Leftrightarrow  A\cdot(x_1,\ldots,x_k)^T=0.
  \end{eqnarray*}
  The matrix $A$ is called a \text{generator matrix} of $P_A$.
\end{definition}
Since $A (x+b) = A x + A  b$, 
we posit that a constraint of an instance of \csp{$P_A$} is of the form $A\cdot (x_1,\ldots,x_k)^T=(b_1,\ldots,b_h)^T$.

As a hard generator matrix, we use a linear code.
A \textit{linear code} of distance $3$ and length $k$ over $\bit$ is a subspace of $\bit^k$ such that every non-zero vector in the subspace has at least $3$ non-zero entries.
We refer to the code below as Hamming code of length $k$.
\begin{fact}
  Let $2^{r-1}-1< k\leq 2^r-1$.
  Then, there exists a linear code of distance $3$ and length $k$ over $\bit$ with dimension $h=k-r$.
\end{fact}
In particular, $|P_A^{-1}(1)|=2^r\leq 2k$ holds for Hamming code $A$.

We define two distributions $\cald_{sat}$ and $\cald_{far}$ of instances of \csp{$P_A$} using Hamming code $A$.
From Lemma~\ref{lmm:expander}, 
for any $\eta>0$ and $d\geq 1$,
there exists $\gamma>0$ such that we have a $k$-uniform $(\gamma,\eta)$-expander $H=(V,E)$ with $n$ vertices and a degree bound $d$.
We use $H$ as an underlying hypergraph of instances generated by $\cald_{sat}$ and $\cald_{far}$.
\begin{itemize}
  \item $\caldsat$: 
    We choose $\bfx_v \in \bit$ for each $v\in V$ uniformly at random. 
    Then, for each edge $e=(v_1,\ldots,v_k)\in E$, 
    we introduce a constraint of the form $A\cdot(x_{v_1},\ldots,x_{v_k})^T=A\cdot(\bfx_{v_1},\ldots,\bfx_{v_k})^T$.
  \item $\caldfar$:
    For each edge $e=(v_1,\ldots,v_k)\in E$,
    we choose $\bfb_e\in \bit^{h}$ uniformly at random and introduce a constraint of the form $A\cdot(x_{v_1},\ldots,x_{v_k})=\bfb_e$.
\end{itemize}
From the construction,
any instance of $\caldsat$ is satisfiable.
On the other hand, 
from Lemma~\ref{lmm:far}, 
for any $\epsilon>0$, by appropriately choosing $d$,
$1-o(1)$ fraction of instances of $\caldfar$ is $(|P^{-1}(0)|/2^k-\epsilon)$-far, i.e., $(1-2k/2^k-\epsilon)$-far.

\begin{theorem}\label{thr:generator-hardness}
  Let $A$ be Hamming code.
  Then, for any $\epsilon>0$, 
  there exists $d\geq 1$ such that every $(1-2k/2^k-\epsilon)$-tester for \csp{$P_A$} with a degree bound $d$ requires $\Omega(n)$ queries.
\end{theorem}
\begin{proof}
  Let $\Phi$ be an instance chosen from $\caldsat$.
  One constraint of \csp{$P_A$} consists of a chunk of $h$ linear equation.
  Thus, in total, there exists $mh$ linear equations in $\Phi$,
  where $m$ is the number of constraints of $\Phi$.
  We can write these equations in the form $Mx = b$ using a matrix $M \in \bit^{mh \times n}$ and a vector $b \in \bit^{mh}$.
  Here, $M$ is uniquely determined by the underlying hypergraph $H$ regardless of $b$.

  We show that for any set of $s\leq \gamma n$ constraints of $\Phi$, 
  the corresponding rows in $M$ are linearly independent.
  Suppose that there exists a set $R$ of $s$ constraints whose corresponding rows are linearly dependent.
  Let $S$ denote the set of variables incident to $R$.
  Since every chunk equation comes from a distance-$3$ code, every linear combination of rows within a chunk must have at least three elements.
  Hence, the linear combination required to derive $0$ must include at least three elements from each of the $s$ constraints.
  To derive $0$, each of these elements must occur an even number of times,
  and hence $s$ constraints can involve at most $ks-3s/2=(k-1-1/2)s$ variables in total.
  If we choose $\eta<1/2$, this is impossible.

  Let $M'x=b'$ be a sub-instance obtained by choosing any $\gamma n$ constraints.
  Since the rows of $M'$ are linearly independent, $M'\bfx$ is also uniformly distributed when $\bfx\in \bit^{n}$ is chosen uniformly at random.
  Thus, no algorithm can distinguish instances of $\caldsat$ from instances of $\caldfar$ with $\gamma n$ queries.
  The theorem follows.
\end{proof}

\begin{proof}[Proof of Theorem~\ref{thr:mis}]
Using a modified version of the FGLSS reduction from \maxkcsp to \mis used in~\cite{Tre01}, 
which is tailored for bounded-degree instances,
we have this theorem.
\end{proof}



\end{document}